\documentclass[11pt]{article}

\usepackage[margin=1in]{geometry}
\usepackage[T1]{fontenc}
\usepackage{amsmath, amssymb, amsthm}
\usepackage{booktabs}
\usepackage{multirow}
\usepackage{graphicx, xcolor}
\usepackage{enumitem}
\usepackage{algorithm}
\usepackage{algpseudocode}
\usepackage{microtype}
\usepackage[numbers,sort&compress]{natbib}
\usepackage{hyperref}
\usepackage{orcidlink}
\let\oldparagraph\paragraph
\renewcommand{\paragraph}[1]{\oldparagraph{#1.}}
\makeatletter
\AtBeginDocument{\renewcommand{\href}[2]{\@ifnextchar.{\@gobble}{}}}
\makeatother

\theoremstyle{plain}
\newtheorem{theorem}{Theorem}[section]
\newtheorem{proposition}{Proposition}[section]
\newtheorem{corollary}{Corollary}[section]
\newtheorem{lemma}{Lemma}[section]
\theoremstyle{definition}
\newtheorem{definition}{Definition}[section]
\newtheorem{remark}{Remark}[section]
\newtheorem{example}{Example}[section]

\newcommand{\E}{\mathbb{E}}
\newcommand{\F}{\mathcal{F}}
\newcommand{\It}{\mathcal{I}_t}
\newcommand{\Ihat}{\hat{\mathcal{I}}}
\newcommand{\norm}[1]{\left\lVert#1\right\rVert}
\DeclareMathOperator{\Var}{Var}

\newcommand{\outlinebox}[1]{\par\noindent #1\par\medskip}

\begin{document}

\title{Estimating the Conditional Forecast-Revision Scale in Sequential Models:
Local-Smoothing Limits, Matched Models, and Cost--Accuracy Trade-offs}

\author{Hui Mean Foo\thanks{\texttt{huimean87@gmail.com}} \qquad
  Yuan-chin Ivan Chang\,\orcidlink{0000-0002-4977-7721}%
  \thanks{Corresponding author. \texttt{ycchang@as.edu.tw}}\\[4pt]
  \normalsize Institute of Statistical Science, Academia Sinica, Taipei, Taiwan}

\date{}

\maketitle

\begin{abstract}
  The \emph{conditional forecast-revision scale}
  $\It=\{\Var(\E[X_{t+1}\mid\F_t]\mid\F_{t-1})\}^{1/2}$ measures the
  history-specific size of the forecast update induced by observing $X_t$.
  Because it is a conditional second moment built from two unknown conditional
  means, it is not directly observed. We study which estimator of $\It$ should
  be used under different structural assumptions and computational budgets. The
  comparison includes a block bootstrap, a conditional-variance model, a fitted
  state-space model, two $O(1)$ streaming smoothers, and the forget gate of an
  already-trained recurrent network. An error decomposition separates
  one-step-prediction error from conditional-second-moment tracking error. We
  show that externally tuned lag-only smoothers can be inconsistent when $\It$
  changes at the sampling scale, although they attain the usual $T^{-2/3}$
  mean-squared-error rate ($T^{-1/3}$ for $\It$) under slow variation; a correctly
  specified state-space estimator escapes
  this limit by using the current state. In volatility-driven designs, a cheap
  conditional-variance model is more accurate and over one hundred times cheaper
  \emph{as a point estimator} than the implemented block bootstrap, whose value
  lies in the sampling distribution it provides rather than in point tracking. In state-driven designs,
  only the structurally matched filter recovers the fast variation. Read directly, a
  trained network's forget gate does not track $\It$ --- though a supervised linear
  probe on the full gate vector does, so $\It$ is linearly decodable from the full
  gate vector but not available for free. These results yield a practical rule: identify
  the conditional-second-moment structure, match the estimator to it, and then
  choose the least costly adequate method.
\end{abstract}



\medskip

\section{Introduction}
\label{sec:intro}

How much does a new observation change the next forecast? We study the
\emph{conditional forecast-revision scale} $\It$, the conditional
root-mean-square size of the one-step forecast update after observing $X_t$.
Large values indicate that the observation materially changes the forecast;
small values indicate that it is largely redundant. This history-specific
quantity is useful wherever an adaptive procedure must decide how strongly to
update, when to increase monitoring sensitivity, or where to spend limited
sampling or computation.

Estimating $\It$ is difficult for two reasons: (1) it is a conditional,
time-varying second moment rather than a realized statistic or a population
constant, and (2) it depends on the difference between two unknown conditional
means, estimated before and after $X_t$ arrives, so the revision can be small
relative to the noise from model fitting. The target is identified through one of
three routes --- shared-history replication, single-path nonparametric estimation
under local stationarity \citep{dahlhaus1997fitting}, or a correctly specified
single-path model --- which imply different estimators and different failure modes.

In this paper, we ask which estimator of $\It$ should be used for a given data
structure and computational budget. We compare a cost-ordered menu: a
block-bootstrap estimator for dependent data
\citep{kunsch1989jackknife,politis1994stationary,lahiri2003resampling}, a
conditional-variance model, a fitted state-space estimator, two streaming
lag-only smoothers, and the forget gate of an already-trained recurrent network.
The comparison is made against exact ground truth in controlled designs and is
supplemented by real-data illustrations.

We organize this paper around a central distinction between forecasting the next observation and tracking the conditional second moment of the resulting revision. Separating these two sources of error (Lemma~\ref{lem:decomp}) clarifies the analysis.
There are four contributions in this paper. First, we define $\It$ as an explicit conditional estimand and give it a decision-theoretic interpretation: the expected one-step reduction in optimal squared prediction risk. We also distinguish three routes to its identification (Sections~\ref{sec:Itproperties}--\ref{sec:regimes}).
Second, we establish a slow--fast smoothing dichotomy. Externally tuned, lag-only smoothers attain mean-squared error $O(T^{-2/3})$ for $\hat v_t=\Ihat_t^2$ (root-mean-squared error $T^{-1/3}$ for $\It$) when $\It$ varies slowly, but they have a nonvanishing error floor for a \emph{predictable} $\It$ that switches at the sampling scale (Theorem~\ref{thm:floor}; Corollary~\ref{cor:dichotomy}). This is a limitation of the estimator class rather than of the estimand: a correctly specified state-space filter recovers state-driven variation (Proposition~\ref{prop:achieve}), whereas a conditional-variance model recovers volatility-driven variation.
Third, we compare structurally distinct estimators against exact ground truth on a common cost--accuracy scale. The results show that structural matching can matter more than computational effort. In the volatility designs studied, the conditional-variance estimator is both more accurate and, as a point estimator, incurs less than one-hundredth of the computational cost of the block bootstrap.
Fourth, we translate this comparison into a diagnostic selection workflow (Section~\ref{sec:selection}) and establish a precise negative result for the forget gate as a \emph{free} proxy. A naive readout fails, although a supervised probe applied to the full gate vector recovers $\It$. Thus, the relevant information is present in the network state but is not directly accessible without additional supervision.

\subsection{Related work and positioning}
\label{sec:related}
Our target sits apart from information-theoretic measures of an observation's
value---predictive information \citep{bialek2001predictability}, expected
information gain \citep{lindley1956measure}, and Bayesian surprise
\citep{itti2009bayesian}: $\It$ is a
history-indexed, squared-loss functional of the conditional-\emph{mean} update, so
it can vanish even when an observation reshapes the predictive variance or tails
(Proposition~\ref{prop:mi}).

When the conditional mean is linear, $\It$ reduces to a conditional volatility,
linking it to ARCH/GARCH models
\citep{engle1982autoregressive,bollerslev1986generalized} and time-varying
volatility \citep{andersen2003modeling,dahlhaus2006statistical}; local stationarity
\citep{dahlhaus1997fitting} supplies a single-path identification regime, and block
and stationary bootstraps handle the dependence
\citep{kunsch1989jackknife,politis1994stationary,lahiri2003resampling}. Model-based
recovery draws on Kalman and regime-switching filters
\citep{kalman1960new,hamilton1989new,douc2004asymptotic}, while gated recurrent
networks offer a different comparison: LSTM and GRU gates are commonly read as
retention or update mechanisms
\citep{hochreiter1997long,gers2000learning,cho2014learning}, yet predictive training
does not force any gate to equal a defined conditional moment.

We are unaware of prior work that makes this history-indexed forecast-revision
scale the primary sequential estimand and jointly studies its identification,
estimator error, computational cost, smoothing limits, and recovery by structurally
matched models; that joint treatment, not a new population identity, is our
contribution.

The remainder of the paper defines $\It$ and its decision-theoretic interpretation (Section~\ref{sec:Itproperties}); develops the three identification routes, the cost-ordered estimator menu, and the smoothing dichotomy (Section~\ref{sec:S3}); evaluates the estimators against exact ground truth, tests the forget-gate proxy, and builds a diagnostic selection workflow (Section~\ref{sec:numerical}); and concludes (Section~\ref{sec:conclusion}). Proofs are in Appendix~\ref{app:proofs} and the online Supplement.

\section{The Conditional Forecast-Revision Scale}
\label{sec:Itproperties}

In this section, we formalize the quantity introduced above as a well-posed sequential estimand. We define $\It$ as the conditional root-mean-square magnitude of the one-step forecast revision and distinguish this history-indexed second moment from a realized statistic or population constant. We then provide its decision-theoretic interpretation, contrast it with information-theoretic measures (Proposition~\ref{prop:mi}), derive the AR(1) case in closed form (Example~\ref{ex:ar1}), and state the identifiability conditions for single-path estimation, motivating the analysis in Section~\ref{sec:S3}.

Write the \emph{forecast revision}, or Doob increment, as
\begin{equation}
D_t \;:=\; \E[X_{t+1}\mid\F_t] \;-\; \E[X_{t+1}\mid\F_{t-1}],
\label{eq:doob-increment}
\end{equation}
which is the increment from time $t-1$ to $t$ of the Doob martingale
$s\mapsto\E[X_{t+1}\mid\F_s]$. By the tower property,
$\E[D_t\mid\F_{t-1}]=0$; hence, $(D_t)$ is a martingale difference sequence and may be viewed as the forecast-revision counterpart of the one-step prediction innovation.

\begin{definition}[Conditional forecast-revision scale]
\label{def:It}
The \emph{conditional forecast-revision scale} (CFRS) $\It$ is the conditional
root-mean-square forecast revision,
\begin{equation}
  \It \;:=\; \bigl(\E[\,\norm{D_t}^2 \mid \F_{t-1}\,]\bigr)^{1/2}.
  \label{eq:It-def}
\end{equation}
\end{definition}

\begin{remark}[Three levels, and why this one]
\label{rem:three-levels}
The increment \eqref{eq:doob-increment} admits three natural norms, nested by
the tower property:
\[
  \underbrace{\norm{D_t}}_{\text{(B) realized}},
  \qquad
  \underbrace{\bigl(\E[\norm{D_t}^2\mid\F_{t-1}]\bigr)^{1/2}}_{\text{(C) conditional }=\,\It},
  \qquad
  \underbrace{\bigl(\E\norm{D_t}^2\bigr)^{1/2}}_{\text{(A) unconditional}},
\]
with $\E[\text{(B)}^2\mid\F_{t-1}]=\text{(C)}^2$ and
$\E[\text{(C)}^2]=\text{(A)}^2$. We adopt (C) for four structural reasons.
\emph{(i)} It is \emph{predictable} ($\F_{t-1}$-measurable), so a gate can in
principle compute it, whereas the realized (B) --- for a Gaussian AR(1),
$D_t=\phi\varepsilon_t$ --- is an i.i.d.\ noise draw, and regularising toward it
would fit noise. \emph{(ii)} It \emph{varies in $t$ for the right reason}: for the
two-state HMM it is large exactly when the filtered state is uncertain, producing
spikes near regime changes, whereas the unconditional (A) is constant under
stationarity. \emph{(iii)} It \emph{reduces to} (A) under conditional
homoskedasticity. \emph{(iv)} It supports an \emph{information-theoretic reading}:
in the Gaussian setting of Proposition~\ref{prop:mi}, $\mathcal J_t$ is strictly
increasing in $\It^2/\sigma^2_{\mathrm{post},t}$, while the realized version admits
no such interpretation ($\It$ alone need not determine the mutual information,
property~(i) below).
\end{remark}

\noindent\textbf{Structural properties.} Four properties fix the scope of the
estimand. \emph{(i) Zero iff no mean update.} $\It=0$ if and only if
$\E[X_{t+1}\mid\F_t]=\E[X_{t+1}\mid\F_{t-1}]$ almost surely: observing $X_t$ does
not move the one-step \emph{mean} forecast. This is a statement about
squared-error mean prediction, \emph{not} about all predictive content --- $X_t$
may sharpen the conditional variance or tails of $X_{t+1}$ while leaving its
conditional mean unchanged, giving $\It=0$ with strictly positive conditional
mutual information. (A pure-ARCH law $X_{t+1}=\sqrt{\omega+\alpha X_t^2}\,
\eta_{t+1}$, $\eta_{t+1}\sim\mathcal N(0,1)$, has $\E[X_{t+1}\mid\F_t]\equiv0$, so
$\It\equiv0$, yet $X_t$ fully determines the scale of $X_{t+1}$.) Thus $\It$
measures the \emph{predictable mean-forecast revision under squared-error loss},
a task-specific quantity; we use the precise name (CFRS) and reserve
``information gain'' as an informal synonym with this scope understood. \emph{(ii) Scale.} $\It$ is homogeneous of degree one in $X$,
and $\It/\sigma$ with $\sigma^2=\Var X_{t+1}$ is scale-free.
\emph{(iii) Gaussian mutual-information link:} Proposition~\ref{prop:mi}.
\emph{(iv) Piecewise stationarity:} under segment-wise conditionally
homoskedastic models with fixed mean-dynamics parameters, $\It$ is constant within
segments and jumps at a change-point $\tau^*$; more generally, stationarity fixes
the distributional law but does \emph{not} make the conditional trajectory
constant --- a stationary GARCH or HMM has a time-varying $\It$ --- so any
monitoring interpretation is model-specific.

\begin{proposition}[History-specific Gaussian information identity]
\label{prop:mi}
Define the history-specific conditional information
\[
  \mathcal J_t := \E\!\Big[\log\frac{p(X_{t+1}\mid X_t,\F_{t-1})}
                                     {p(X_{t+1}\mid\F_{t-1})}\ \Big|\ \F_{t-1}\Big].
\]
Suppose the conditional law of $(X_t,X_{t+1})$ given $\F_{t-1}$ is jointly
Gaussian and $\sigma^2_{\mathrm{post},t}:=\Var(X_{t+1}\mid\F_t)$ is
$\F_{t-1}$-measurable. Then, almost surely,
$\mathcal J_t = \tfrac12\log(1+\It^2/\sigma^2_{\mathrm{post},t})$, and the
conventional scalar conditional mutual information satisfies
$I(X_t;X_{t+1}\mid\F_{t-1})=\E[\mathcal J_t]$. If the post-update variance depends
on $X_t$, $\It$ alone determines neither quantity.
\end{proposition}

\begin{proof}[Proof idea]
The law of total variance and the $\F_{t-1}$-measurability of
$\sigma^2_{\mathrm{post},t}$ give
$\Var(X_{t+1}\mid\F_{t-1})=\It^2+\sigma^2_{\mathrm{post},t}$, so the post- and
pre-update laws are Gaussian with variances $\sigma^2_{\mathrm{post},t}$ and
$\It^2+\sigma^2_{\mathrm{post},t}$; the conditional-mean log-density ratio then
collapses to one half the log of their ratio,
$\mathcal J_t=\tfrac12\log\!\big(1+\It^2/\sigma^2_{\mathrm{post},t}\big)$, and a
further expectation gives $I(X_t;X_{t+1}\mid\F_{t-1})=\E[\mathcal J_t]$. When
$\sigma^2_{\mathrm{post},t}$ depends on $X_t$ this decomposition breaks: a
variance-only channel has $\It=0$ with $\mathcal J_t>0$, the stated counterexample.
The full computation, term by term, is in Supplement~S11.
\end{proof}

\begin{example}[AR(1): closed form, and a degenerate case]
\label{ex:ar1}
For $X_{t+1}=\phi X_t+\varepsilon_{t+1}$ with i.i.d.\ innovations,
$D_t=\phi X_t-\phi^2X_{t-1}=\phi\,\varepsilon_t$, hence
$\It=|\phi|\,\sigma_\varepsilon$. This gives an exact closed form against which any
estimator can be checked ---
but note that it is \textbf{constant in $t$}. AR(1) is therefore a
\emph{level-accuracy} benchmark only; it carries no time variation, so it cannot
be used to validate that an estimator (or a gate) tracks the \emph{trajectory}
of $\It$. Time-variation studies must use the HMM, regime-switching, or
conditionally heteroskedastic processes, where the conditional second moment
genuinely moves.
\end{example}

\paragraph{Connection to predictive power loss}
$\It$ has an exact decision-theoretic reading under squared loss. With
$m_s=\E(X_{t+1}\mid\F_s)$, let $R_{t-1}:=\E[(X_{t+1}-m_{t-1})^2\mid\F_{t-1}]$ and
$R_{t\mid t-1}:=\E[(X_{t+1}-m_t)^2\mid\F_{t-1}]$ be the pre-update Bayes risk and
the previsible expectation of the post-update Bayes risk. Orthogonality of
conditional expectations gives
$R_{t-1}-R_{t\mid t-1}=\E[(m_t-m_{t-1})^2\mid\F_{t-1}]=\It^2$.
Thus $\It^2$ is exactly the expected one-step reduction in optimal squared
prediction risk generated by revealing $X_t$. This identity is more general than
the Gaussian mutual-information link (Proposition~\ref{prop:mi}) and is the
primary meaning we use throughout.

\paragraph{Identifiability and approximation error}
When a single fitted one-step predictor $\hat m$ supplies both forecasts,
$\hat D_t-D_t=\varepsilon_t-\varepsilon_{t-1}$ with $\varepsilon_s=\hat m_s-m_s$,
so the revision error depends on the \emph{change} in predictor error between two
adjacent information sets, not on either absolute error. In conditional $L^2$,
$\norm{\hat D_t-D_t}_{t-1}\le\norm{\varepsilon_t}_{t-1}+\norm{\varepsilon_{t-1}}_{t-1}$,
and any common fitted-model component shared by $\varepsilon_t$ and
$\varepsilon_{t-1}$ cancels in the difference. This motivates the sharper
empirical question posed later (Section~\ref{sec:decomp}): once the true predictor
is supplied, how much error remains in estimating the conditional second moment?

\section{Estimators and Their Limits}
\label{sec:S3}

The estimand $\It$ is \emph{predictable and time-varying}
(Definition~\ref{def:It}); estimating it is qualitatively harder than estimating
an unconditional moment, and both the difficulty and the shape of the menu
depend on what data are available.

\subsection{Three identification routes}
\label{sec:regimes}
Because $\It$ is $\F_{t-1}$-measurable, sample size alone does not determine what
is identified; three routes, using different sources of information, must be kept
apart.

\textbf{Route 1: shared-history replication.} When the past can be held fixed and
the future replicated --- controlled simulation, or a designed experiment --- the
cross-branch mean of $D_t^2$ estimates the history-specific $v_t=\It^2$ exactly as
the number of branches $\to\infty$. This is how the exact ground truth is obtained
in the synthetic studies. \emph{Independent} paths with distinct histories (an
i.i.d.\ panel) do \emph{not} identify this quantity by naive pooling --- they
identify the unconditional level of Remark~\ref{rem:regimes-estimate} --- though a
conditional mapping is recoverable if the histories or a sufficient state are
modelled.

\textbf{Route 2: nonparametric single-path.} With one series and no model, a lag
window substitutes nearby times for unavailable replications
($w\to\infty$, $w/T\to0$); this \emph{requires} the conditional moment to be
approximately stable over the window --- local stationarity, recurrence in a
low-dimensional state, or another explicit smoothness condition
(Theorem~\ref{thm:bbconsistency}). It is exactly this route that the lag-smoother
limit of Section~\ref{sec:limit} constrains when $\It$ switches at the sampling
scale.

\textbf{Route 3: model-based single-path.} A correctly specified, identifiable
parametric or semiparametric model --- a GARCH or a finite-state filter ---
recovers the path-specific $\It$ from a single \emph{stationary, ergodic} path
with \emph{no} local-stationarity assumption: the whole series identifies the
finite-dimensional parameter and the fitted filter returns $\It$ pointwise
(Proposition~\ref{prop:achieve}). Local stationarity is thus a property of Route~2,
not a precondition for every single-path estimator.

This split organizes the menu: the windowed surrogates and the block bootstrap
are Route-2 estimators, while the conditional-variance and state-space estimators
are Route-3 estimators.

\subsection{Every estimator is a predictor plus a second-moment estimator}
\label{sec:decomp}

Each estimator of $\It$ has two ingredients: a one-step predictor, which forms
the forecast revision, and a conditional-second-moment estimator, which turns
that revision into $\It$. The following elementary bound separates their errors
and governs the whole menu.

\begin{lemma}[Error decomposition]
\label{lem:decomp}
For $s\in\{t-1,t\}$ let $\hat m_s$ be a square-integrable predictor of
$X_{t+1}$, with error $\varepsilon_s=\hat m_s-m_s$, $m_s=\E[X_{t+1}\mid\F_s]$.
Let $\Ihat_t$ be \emph{any} estimator built from the estimated revision
$\hat D_t=\hat m_t-\hat m_{t-1}$, and write its second-moment error as
$\delta_t=\bigl|\Ihat_t-(\E[\hat D_t^2\mid\F_{t-1}])^{1/2}\bigr|$ and the
conditional norm as $\norm{\cdot}_{t-1}=(\E[\cdot^2\mid\F_{t-1}])^{1/2}$. Then
\begin{equation}
  \bigl|\Ihat_t-\It\bigr|
  \;\le\; \delta_t \;+\; \norm{\varepsilon_t}_{t-1}
          \;+\; \norm{\varepsilon_{t-1}}_{t-1}.
  \label{eq:decomp}
\end{equation}
\end{lemma}

\begin{proof}
$\hat D_t-D_t=\varepsilon_t-\varepsilon_{t-1}$, so by the reverse triangle
inequality for the conditional $L^2$ norm,
$\bigl|\norm{\hat D_t}_{t-1}-\norm{D_t}_{t-1}\bigr|
 \le\norm{\hat D_t-D_t}_{t-1}
 \le\norm{\varepsilon_t}_{t-1}+\norm{\varepsilon_{t-1}}_{t-1}$;
since $\It=\norm{D_t}_{t-1}$, adding $\delta_t$ through the triangle inequality
gives~\eqref{eq:decomp}.
\end{proof}

\noindent The bound uses only square-integrability of $\hat m_{t-1},\hat m_t$;
$\F_s$-measurability is needed separately for the \emph{online}, predictable
reading, and the offline full-sample fits used below satisfy it only after
conditioning on the fitted parameters (the offline/online distinction made
precise in Section~\ref{sec:ssm}). The two terms have different remedies, and
separating them is what
makes the problem tractable. The second-moment error $\delta_t$ is what the
estimators below are engineered to control (and what the bootstrap consistency
theorem addresses); crucially, it depends not only on sample size but on how fast
$\It$ itself varies --- a windowed or block estimator of a conditional second
moment trades bias against variance through its bandwidth, and its bias grows
with the rate of change of that moment, so it can track a slowly-varying $\It$
but not one that jumps with a latent state. The predictor terms
$\norm{\varepsilon_s}_{t-1}$ are inherited from the one-step predictor: if the
predictor is $L^2$-consistent they vanish, but a predictor restricted to a class
that cannot represent $m_s$ --- an affine predictor under a nonlinear conditional
mean --- leaves a floor
$\inf\norm{\varepsilon_s}\ge\operatorname{dist}_{L^2}(m_s,\mathcal{H})>0$. Both
terms can bind, and \emph{which} one does is an empirical question we answer
directly: re-running each estimator with the \emph{true} conditional mean
supplied (an oracle-predictor control, Section~\ref{sec:sim}) sets
$\varepsilon_s=0$ and isolates $\delta_t$. On the state-driven processes that
control barely helps --- so the binding term is $\delta_t$, not the predictor ---
and this dictates the remedy: not a better predictor, but a second-moment
estimator whose \emph{model} is matched to the mechanism generating $\It$, which
is what the conditional-variance and state-space estimators below supply.

\subsection{A Fundamental Limit for Local Estimators}
\label{sec:limit}

The oracle-predictor finding says the binding error is $\delta_t$, the
second-moment-tracking term. This subsection makes that precise: for the entire
class of estimators that track $\It$ by \emph{externally weighted local
averaging} of the realised second moment --- including the windowed RMS, the
trailing-RMS implementation of the innovation-scaled surrogate, and the point
estimate underlying the block bootstrap --- there is an irreducible
bias--variance floor that, for the predictable sampling-scale switching
construction below, forbids consistency at every bandwidth. Write
$v_t:=\It^2=\E[D_t^2\mid\F_{t-1}]$ for the estimand.

\begin{definition}[Externally tuned lag-only convex smoother]
\label{def:local}
With $W_t=\{t-w,\dots,t-1\}$, an estimator
$\hat v_t=\Ihat_t^2=\sum_{s\in W_t}a_{t,s}D_s^2$ is an \emph{externally tuned
lag-only convex smoother} when $a_{t,s}\ge0$, $\sum_s a_{t,s}=1$, and the weights
are deterministic or measurable with respect to an external tuning field
$\mathcal T$ independent of the evaluation variance path and innovations (for the
switching corollary the weights may be $\F_{t-2}\vee\mathcal T$-measurable, but must
not use the newly observed state at time $t-1$); its effective width is
$w_t^{\mathrm{eff}}=(\sum_s a_{t,s}^2)^{-1}\le w$. Uniform windows, fixed kernels,
and the point estimate underlying the rolling block bootstrap are covered, as are
both lightweight surrogates (A and B) when their bandwidth is fixed externally. We do
\emph{not} claim the theorem for weights learned from the same $D_s^2$ --- their
dependence on the innovation terms would invalidate the variance decomposition
below --- nor for estimators that infer a latent state nonlinearly
(Proposition~\ref{prop:achieve} shows those can escape the bound).
\end{definition}

\noindent We take the variance path $\{v_r\}$ to be \emph{exogenous}: it is
generated independently of the standardised innovations $\{\xi_s\}$ in the revision
model $D_s=\sqrt{v_s}\,\xi_s$, with $\xi_s$ i.i.d., mean zero, unit variance, of
finite fourth moment, and independent of the tuning field $\mathcal T$. Then
$\E[D_s^2\mid\F_{s-1}]=v_s$, the conditional variance is nondegenerate and
two-sided, $\kappa_0 v_s^2\le\Var(D_s^2\mid\F_{s-1})\le\kappa_1 v_s^2$
($0<\kappa_0\le\kappa_1$; $\kappa_0=\kappa_1=2$ for conditionally Gaussian $D_s$), and ---
the property the proof uses --- conditioning on the whole path $\{v_r\}$ and
$\mathcal T$ leaves
$\E[\eta_s\mid\{v_r\},\mathcal T]=0$ with the excess terms $\eta_s:=D_s^2-v_s$ conditionally
uncorrelated. Exogeneity holds for the oracle revision model and for the
observed-switching construction of Corollary~\ref{cor:dichotomy}(ii) (there
$\zeta\perp\xi$). It \emph{fails} when the variance path is itself a function of
past squared revisions, as in GARCH, where $D_s^2$ is recoverable from $\{v_r\}$
(e.g.\ $v_{s+1}=\omega+\alpha D_s^2+\beta v_s$); the theorem then bounds the
lag-only \emph{class} rather than those processes directly, and the empirical wall
on them (Section~\ref{sec:sim}) is documented numerically --- consistent with, but
not implied by, the theorem. Let $0<v_{\min}\le v_s\le v_{\max}$.

\begin{theorem}[Bias--variance identity for lag-only smoothing]
\label{thm:floor}
For every estimator of Definition~\ref{def:local}, conditional on the variance
path $\{v_r\}$ and the tuning field $\mathcal T$,
\[
  \E\!\big[(\hat v_t-v_t)^2\mid\{v_r\},\mathcal T\big]
  = B_t^2 + \sum_{s\in W_t} a_{t,s}^2\,\Var(D_s^2\mid\F_{s-1}),
  \qquad B_t:=\sum_{s\in W_t} a_{t,s}(v_s-v_t),
\]
and consequently
\[
  B_t^2 + \frac{\kappa_0 v_{\min}^2}{w_t^{\mathrm{eff}}}
  \;\le\;
  \E\!\big[(\hat v_t-v_t)^2\mid\{v_r\},\mathcal T\big]
  \;\le\;
  B_t^2 + \frac{\kappa_1 v_{\max}^2}{w_t^{\mathrm{eff}}}.
\]
\end{theorem}

\noindent The two terms pull against each other through the bandwidth: shrinking
$w_t^{\mathrm{eff}}$ to cut the smoothing bias $B_t$ inflates the variance, and
vice versa. Whether they can be made small \emph{together} depends entirely on how
fast $v_s$ moves.

\begin{corollary}[Slow variation and predictable sampling-scale switching]
\label{cor:dichotomy}
\emph{(i) Slow variation.} If $v_s=g(s/T)$ with $g$ Lipschitz and the uniform
window is used, then $B_t=O(w/T)$ and
$\E[(\hat v_t-v_t)^2]=O(w^2/T^2+1/w)$; choosing $w\asymp T^{2/3}$ gives
mean-squared error $O(T^{-2/3})$ for $v_t=\It^2$. Since $v_t$ is bounded away from
zero, the delta method gives the same order for $\Ihat_t$ and root-mean-square
error $O(T^{-1/3})$ for $\It$. This is achievability for the uniform window, not a
claim that \emph{every} convex smoother attains the rate.

\emph{(ii) Predictable switching.} Let $\{\zeta_t\}$ be an \emph{observed}
stationary symmetric Markov chain on $\{-1,1\}$ with
$\E(\zeta_t\mid\zeta_{t-1})=\rho\zeta_{t-1}$, $|\rho|<1$; let $\xi_t$ be i.i.d.\
(mean $0$, variance $1$), independent of $\zeta$; and set
$v_t=\mu+\delta\zeta_{t-1}$, $D_t=\sqrt{v_t}\,\xi_t$ ($\mu>\delta>0$), on
$\F_t=\sigma(\zeta_r,\xi_r:r\le t)$. Then $v_t$ is $\F_{t-1}$-measurable, so the
target is genuinely predictable; yet for every externally tuned lag-only convex
smoother, $\E[(\hat v_t-v_t)^2]\ge\delta^2(1-\rho^2)>0$ for all $T$ and all
bandwidths.
The restricted class cannot consistently track even a \emph{predictable} variance
path that switches at the sampling scale.
\end{corollary}

\noindent The floor comes from the one-step delay in lag-only squared-revision
averages: because $s\le t-1$, the estimable part $\sum_s a_{t,s}\zeta_{s-1}$ is
$\F_{t-2}\vee\mathcal T$-measurable, whereas $v_t$ depends on the newly observed $\zeta_{t-1}$,
so the mean-squared error is at least
$\delta^2\,\E\{\Var(\zeta_{t-1}\mid\F_{t-2})\}=\delta^2(1-\rho^2)$. This is a lower
bound for this specific class --- \emph{not} a claim that every local or
nonparametric method is impossible.

Figure~\ref{fig:dichotomy} confirms both parts of Corollary~\ref{cor:dichotomy}
numerically in the oracle model $D_s=\sqrt{v_s}\,\xi_s$. Panel~(a) plots the
trailing-window mean-squared error against bandwidth at $T=1600$: under slow
variation a bandwidth trades bias against variance and drives the error down,
whereas under predictable sampling-scale switching the error settles onto the
floor $\delta^2(1-\rho^2)$ at every bandwidth. Panel~(b) plots the optimal-window
error against series length: the slow path decays at the predicted $T^{-2/3}$ order
(fitted slope $-0.62$), while the switching path is essentially flat (slope
$-0.02$) --- the empirical signature of a floor that neither bandwidth nor sample
size removes.

\begin{figure}[t]
\centering
\includegraphics[width=0.78\linewidth]{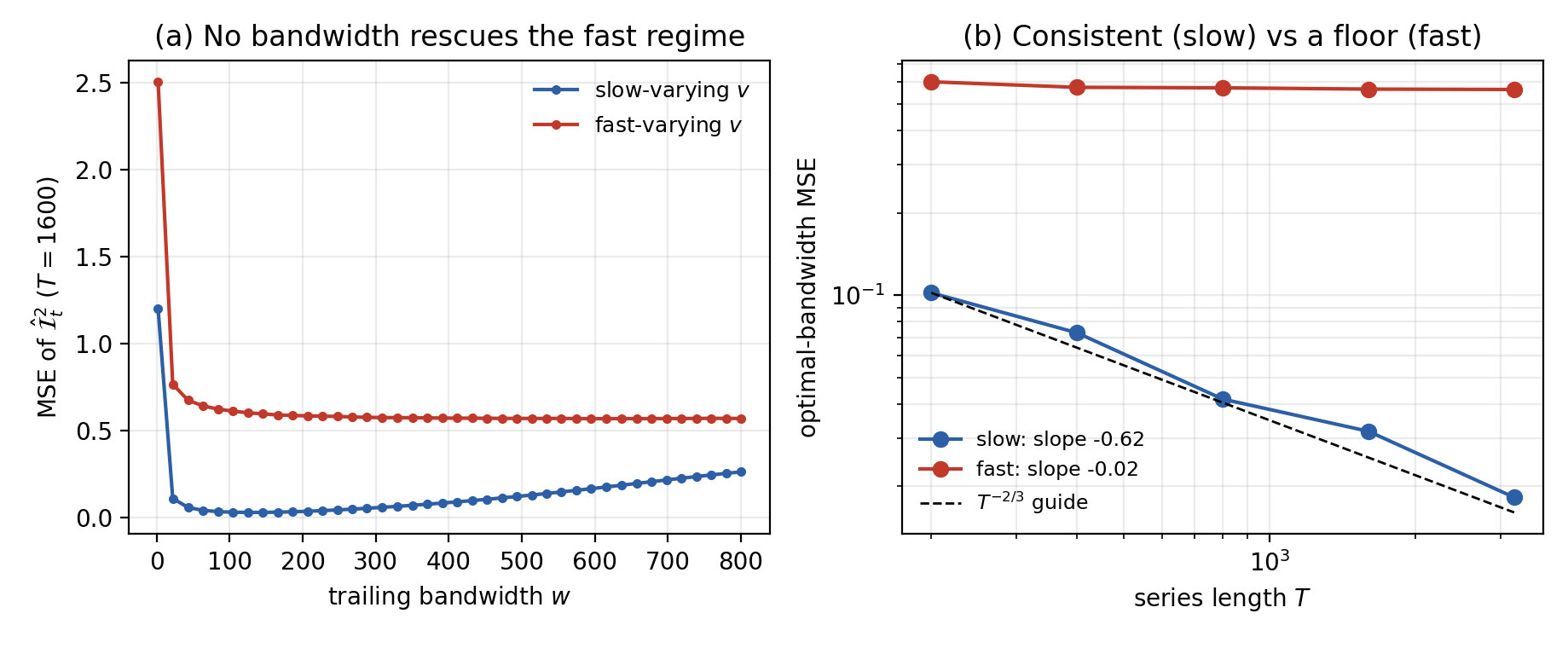}
\caption{Illustration of Corollary~\ref{cor:dichotomy} in the oracle revision
model $D_s=\sqrt{v_s}\,\xi_s$. \emph{(a)} MSE of the uniform trailing-window
estimator versus bandwidth at $T=1600$: for a slowly-varying $v$ a bandwidth
drives the error down, but for a predictable sampling-scale switching $v$ the
error floors regardless of bandwidth. \emph{(b)} Optimal-window MSE versus series
length: the slow path follows the predicted $T^{-2/3}$ order (fitted slope
$-0.62$) while the switching path retains a floor (slope $-0.02$). The figure
illustrates a limit of lag-only averaging, not an impossibility for nonlinear or
state-adaptive estimators.}
\label{fig:dichotomy}
\end{figure}

\begin{remark}[How a matched model can escape]
\label{rem:escape}
Theorem~\ref{thm:floor} applies only to estimators linear in past squared
revisions with externally chosen weights. An estimator that reads the current
state --- the observed $\zeta_{t-1}$ in the switching construction, or a filtered
distribution in latent-state applications (Sections~\ref{sec:cvestimator},
\ref{sec:ssm}) --- is a nonlinear functional and lies outside the theorem. Its
success is not automatic: Proposition~\ref{prop:achieve} requires correct
specification, and Supplement~S2 quantifies the deterioration under
misspecification. The lower-bound and achievability constructions are related but
not identical; a matched lower-and-upper result for one latent-state model remains
open.
\end{remark}

\subsection{The Block-Bootstrap Estimator}

The block bootstrap \citep{kunsch1989jackknife,politis1994stationary} targets the
\emph{sampling distribution} of the trailing window second moment. Because
averaging over resamples estimates the conditional second moment of
Definition~\ref{def:It} rather than a single realized revision, the bootstrap is a
natural estimator when the target is a moment; its role, made precise below, is
uncertainty quantification rather than a distinct point estimate.

\begin{algorithm}[ht]
\caption{Rolling Block-Bootstrap Estimation of $\It$ (point estimate $+$ band)}
\label{alg:bootstrap}
\begin{algorithmic}[1]
\Require Series $x_1,\ldots,x_T$; window $w$; block length $\ell$; $B$ replications
\State Fit the one-step predictor \emph{once} on the whole series; form residuals
  $\hat e_s=x_s-\hat m(\F_{s-1})$ and slope $\hat\phi$, so that the forecast
  revision is $\hat D_s=\hat\phi\,\hat e_s$
  \Comment{fitting is \emph{not} redone inside the resampling loop}
\For{$t=w+1,\ldots,T$}
  \State $W_t \leftarrow \{t-w,\ldots,t-1\}$
  \For{$b = 1,\ldots,B$}
    \State Draw a circular block-bootstrap resample $\{\hat e^{(b)}_s\}$ of
      $\{\hat e_s: s\in W_t\}$ with block length $\ell$
    \State $S_b \leftarrow \hat\phi^{2}\,w^{-1}\sum_{s\in W_t}(\hat e^{(b)}_s)^2$
  \EndFor
  \State $\hat v_t^{\mathrm{boot}} \leftarrow B^{-1}\sum_b S_b$;\quad
    $[L_v,U_v]\leftarrow$ empirical quantiles of $\{S_b\}$
    \Comment{variance ($\It^2$) scale}
  \State $\Ihat_t^{\mathrm{boot}} \leftarrow \sqrt{\hat v_t^{\mathrm{boot}}}$;\quad
    $[L_I,U_I]\leftarrow[\sqrt{L_v},\,\sqrt{U_v}]$
    \Comment{scale of $\It$}
\EndFor
\Ensure $\hat v_t^{\mathrm{boot}}$ with band $[L_v,U_v]$ (variance scale) and
  $\Ihat_t^{\mathrm{boot}}$ with band $[L_I,U_I]$ (scale of $\It$) for each $t$
\end{algorithmic}
\end{algorithm}

\noindent The predictor is fitted once and then held fixed: block-resampling the
\emph{residuals} within the trailing window preserves the fitted revision
structure, whereas re-fitting on a block-shuffled series would corrupt
$\hat\phi$. (The fit is \emph{offline}, on the full sample, so this is
retrospective reconstruction, kept separate from the online,
$\F_{t-1}$-measurable use motivated in the introduction.) As $B\to\infty$ the point
estimate $(\Ihat_t^{\mathrm{boot}})^2$ converges to the trailing window second
moment $\hat\phi^2 w^{-1}\sum_{s\in W_t}\hat e_s^2$ --- so the bootstrap is
\emph{not} a distinct point estimator: additional replication reduces
Monte-Carlo error and supports uncertainty quantification, but cannot remove the
temporal smoothing bias inherited from the underlying local statistic. This is the
honest reading of the ``accurate, expensive'' tier: pay for it when a sampling
distribution or interval is wanted, not for point accuracy.

\begin{theorem}[Consistency and bootstrap validity in the locally-stationary regime]
\label{thm:bbconsistency}
  Consider a triangular array $(X_{t,T})$ whose forecast-revision conditional
  second moment is \emph{locally stationary}: there is a Lipschitz
  $G:[0,1]\to(0,\infty)$ with $\E[D_{t,T}^2\mid\F_{t-1}]=G(t/T)$, so that
  $\It=\sqrt{G(t/T)}$. Assume:
  \begin{enumerate}[label=\emph{(A\arabic*)},leftmargin=*,itemsep=0pt,topsep=2pt]
    \item $\sup_{t,T}\E|D_{t,T}|^{4+\delta}<\infty$ for some $\delta>0$;
    \item $(D_{t,T}^2)$ is $\alpha$-mixing with
      $\sum_{k\ge1}\alpha(k)^{\delta/(4+\delta)}<\infty$, a standard dependence
      condition used in block-bootstrap central limit theorems
      \citep{kunsch1989jackknife,lahiri2003resampling};
    \item the fitted one-step predictor is uniformly $L^2$-consistent over the
      estimation window,
      $\max_{s\in W_t}\norm{\hat m_s-m_s}_{L^2}=o(1)$.
  \end{enumerate}
  We state the result for the linear residual construction $\hat D_s=\hat\phi\,
  \hat e_s$ of Algorithm~\ref{alg:bootstrap}; a general $\hat D_s$ requires the
  analogous conditions on $\hat D_s^2$.
  Write $\hat v_t^{\mathrm{loc}}=w^{-1}\sum_{s\in W_t}\hat D_s^2$ for the local
  point estimate and $\bar v_{t,T}=w^{-1}\sum_{s\in W_t}G(s/T)$. With trailing
  window $W_t=\{t-w,\dots,t-1\}$ and block length $\ell$ satisfying $w\to\infty$,
  $w/T\to0$, $\ell\to\infty$, $\ell/w\to0$, $B\to\infty$, and at each rescaled time
  $u=t/T\in(0,1)$:
  \emph{(i) Consistency.} $\hat v_t^{\mathrm{loc}}\xrightarrow{p}G(u)$, and hence
  $\Ihat_{t,T}^{\mathrm{boot}}\xrightarrow{p}\It$.
  \emph{(ii) Bootstrap validity (conditional on the fitted predictor).} Suppose, in
  addition, that the mixing rate yields the moving-block bootstrap central limit
  theorem for the local mean and that predictor estimation is $\sqrt w$-negligible,
  $\sqrt w\,\max_{s\in W_t}\norm{\hat m_s-m_s}_{L^2}\to0$ (so the predictor does not
  enter at the $\sqrt w$ scale). Then, treating the fitted $(\hat\phi,\hat e_s)$ as
  fixed under resampling and writing $\hat v_t^*$ for the bootstrap statistic,
  \[
  \begin{aligned}
    \sup_x\Big|&
      \Pr\!{}^*\!\big\{\sqrt w\,(\hat v_t^*-\hat v_t^{\mathrm{loc}})\le x\big\}\\
      &{}-\Pr\big\{\sqrt w\,(\hat v_t^{\mathrm{loc}}-\bar v_{t,T})\le x\big\}
    \Big| \xrightarrow{p}0 .
  \end{aligned}
  \]
\end{theorem}

\noindent The bootstrap thus reproduces the \emph{sampling uncertainty} of the
local average, not a distinct point estimate; the number of draws $B$ controls
only Monte-Carlo error. A confidence interval for $v_t$ itself is valid only once
the smoothing bias $\bar v_{t,T}-G(u)=O(w/T)$ is negligible relative to the
interval width $O(w^{-1/2})$ --- that is, $w=o(T^{2/3})$, one must \emph{under}-smooth
relative to the MSE-optimal bandwidth of Corollary~\ref{cor:dichotomy}(i) --- or is
explicitly bias-corrected. The reported band is a percentile interval on the
$\Ihat_t^2$ (variance) scale; a band for $\It$ follows by taking $\sqrt{\cdot}$ of
its endpoints.

\begin{remark}[What the other regimes estimate]
\label{rem:regimes-estimate}
The conditioning is essential. Under independent replication with \emph{distinct}
histories (an i.i.d.\ panel) a cross-replication root-mean-square converges to the
\emph{unconditional} level~(A) of Remark~\ref{rem:three-levels}, not the
path-specific $\It$; the conditional target is recovered by replication only when
the past is held fixed (branching or Monte-Carlo replication). For a \emph{fixed}
stationary process ($G$ constant) the
trailing-window mean converges to $\E[D^2]=\It^2$ --- consistent, but with no
trajectory to track (a level check, Example~\ref{ex:ar1}). The interesting case is
a varying $G$: slowly on the $T$-scale it is nonparametrically estimable at rate
$T^{-2/3}$ (Theorem~\ref{thm:bbconsistency}, Corollary~\ref{cor:dichotomy}(i))
\citep{dahlhaus1997fitting,dahlhaus2006statistical}; at the sampling scale it is
not (Corollary~\ref{cor:dichotomy}(ii)). This corrects the na\"ive ``fixed $t$,
$T\to\infty$'' reading, which conflates the constant level with the trajectory.
\end{remark}

\paragraph{Block-length sensitivity}
The block length is a further practical choice. Because it affects the estimated
sampling distribution, not the point estimate, it is treated as a sensitivity
parameter --- swept over a stable grid and assessed by the long-run variance of
$\hat D_s^2$ rather than by minimising point-estimate MSE, with the point estimate
unchanged across the grid; details are in the Supplementary Material (Section~S6).

\subsection{Two Lightweight Alternatives}

Both surrogates target the \emph{conditional} second moment of
Definition~\ref{def:It}, not a single realized revision, and that distinction is
what makes them nontrivial. A lone difference $|D_t|$ is one draw whose conditional
mean square \emph{is} the estimand, so it cannot serve as an estimator on its own.
Each surrogate therefore replaces the bootstrap's average over replications with an
average over a short \emph{local window}, buying an $O(1)$ per-step cost at the
price of conditional resolution --- the same bias--resolution trade-off that
Theorem~\ref{thm:floor} makes precise.

\paragraph{Surrogate A: windowed RMS of forecast revisions}
Surrogate~A is the root-mean-square revision of a $k$-window running predictor over
a trailing window of width $w$,
$\Ihat_{t,A} = \bigl(w^{-1}\sum_{s=t-w+1}^{t}(\bar{x}_{s,k}-\bar{x}_{s-1,k})^2\bigr)^{1/2}$
($1\ll k,\ w\ll T$).

\outlinebox{$O(1)$ per step with a rolling accumulator, hence streaming. Two biases
compete: the predictor window $k$ (mean approximation) and the averaging window $w$
(a local stand-in for conditioning on $\F_{t-1}$); larger $w$ cuts variance but
blurs the time variation the estimator exists to capture. Supplement~S4 gives the
exact decomposition.}

\paragraph{Surrogate B: innovation-scaled surrogate}
Surrogate~B is $\Ihat_{t,B} = |\hat\phi|\cdot\hat\sigma_{t}$, where $\hat\sigma_t$
is a conditional scale estimate for the one-step innovation (a trailing RMS of
$\hat{e}_t$, or a fitted conditional-variance model) and $\hat\phi$ the estimated
first-order dependence.

\outlinebox{It is motivated by Example~\ref{ex:ar1}: since $D_t=\phi\varepsilon_t$,
Definition~\ref{def:It} gives $\It = |\phi|\,\sigma_t$, so the correct plug-in uses
the \emph{conditional scale} of the innovation, not its realized value. Under
homoskedasticity the surrogate is constant --- the honest answer for AR(1) --- and
becomes informative only when innovations are conditionally heteroskedastic. The
identity extends to AR($p$) ($\It=|\phi_1|\sigma_t$) and, for a VAR with first-lag
matrix $A_1$ and innovation covariance $\Sigma_t$, gives
$D_t=A_1\varepsilon_t$ and
$\It^2=\operatorname{tr}(A_1\Sigma_t A_1^\top)$.}

\subsection{A Model-Based Estimator: Conditional Variance}
\label{sec:cvestimator}
\outlinebox{When the revision is driven by conditional heteroskedasticity --- the
leading case of smoothly-varying $\It$ --- a fitted conditional-variance model
estimates $\It$ directly: with $\hat\sigma_{t\mid t-1}$ the conditional standard
deviation from a GARCH-type model
\citep{engle1982autoregressive,bollerslev1986generalized}, the identity of
Example~\ref{ex:ar1} gives $\Ihat_t^{\mathrm{cv}} = |\hat\phi_t|\,\hat\sigma_{t\mid
t-1}$, where $\hat\phi_t$ is the fitted first-order sensitivity of the conditional
mean. It is lower-variance and cheaper than the bootstrap when well specified,
biased when it is not (e.g.\ nonlinear predictability with near-homoskedastic
noise), and places $\It$ in the realized-volatility literature
\citep{andersen2003modeling}, the natural scholarly home of the estimand.}

\subsection{A Matched State-Space Estimator}
\label{sec:ssm}

The conditional-variance estimator matches the model to the \emph{volatility}
structure of $\It$. When $\It$ is instead \emph{state-driven} --- its variation
comes from uncertainty about a latent state, so the conditional moment jumps as
the state posterior sharpens or diffuses --- the matched model is a state-space
one. We fit a two-state Markov-switching AR(1) --- $S_t\in\{1,2\}$ Markov($P$),
$X_t = \mu_{S_t} + \phi_{S_t} X_{t-1} + \sigma_{S_t}\,\varepsilon_t$ --- by the EM
(Baum--Welch) algorithm, and read $\It$ off its fitted filter --- the
discrete-state analogue of the Kalman filter \citep{kalman1960new} --- exactly
as the true $\It$ is computed (Section~\ref{sec:Itproperties}), but with fitted
parameters in place of the oracle ones. This one model nests both state-driven
benchmarks: $\phi_k\equiv 0$ is the Gaussian-emission HMM, $\mu_k\equiv 0$ the
regime-switching AR. The two structures have very different likelihood geometry
--- a shift in level versus a shift in the AR coefficient alone --- and a single
mean-separated start converges to a spurious optimum on the AR-only case, so we
run a few fixed and random initialisations and keep the highest-likelihood fit.

\outlinebox{This is the dearest entry (multiple EM passes), and deliberately so:
the general tool one reaches for once the cheap estimators fail, i.e.\ once $\It$
is diagnosed as state-driven. It is accurate in \emph{both} regimes, but on
volatility-driven $\It$ the cheap conditional-variance estimator
Pareto-dominates it as a point estimator, so one pays for it only when the
structure demands. Its role is to show the ``wall'' of Section~\ref{sec:sim} is a
property of the estimator, not of $\It$: matched to the right structure, $\It$ is
recovered where every generic estimator, and even an oracle predictor, reads
zero.}

\begin{proposition}[Achievability by a matched model]
\label{prop:achieve}
Fix a rescaled time $u\in(0,1)$ and set $t=\lfloor uT\rfloor$. Suppose the
observed process is stationary and ergodic and lies in the finite-state
Markov-switching autoregressive family that $\Ihat_t^{\mathrm{ssm}}$ fits, with
parameter $\theta_0$ identifiable up to label permutation and satisfying the
regularity conditions of Douc et al. \citep{douc2004asymptotic}. Let
$\hat\theta$ be the offline
maximum-likelihood estimate from the whole series and
$v_t(\theta)=\Var(D_t\mid\F_{t-1};\theta)$ the forecast-revision variance returned
by the Hamilton filter run to time $t$. Then
$\Ihat_t^{\mathrm{ssm}}=\sqrt{v_t(\hat\theta)}\xrightarrow{p}\It$ as $T\to\infty$,
even when $\It$ varies at the sampling scale.
\end{proposition}

\begin{proof}[Proof outline]
The argument has four steps; the technical lemmas and the regularity conditions
(SS1)--(SS4) are given in full in Supplement~S10.
\emph{(1) Consistency.} The offline MLE is strongly consistent up to label
permutation \citep{douc2004asymptotic}, and $v_t$ --- a functional of the
observable predictive law --- is permutation-invariant.
\emph{(2) Filter forgetting.} A strictly positive transition matrix makes the
Hamilton filter a Birkhoff--Hopf contraction in the Hilbert projective metric, so
it forgets its initialisation geometrically and $v_t(\theta)$ agrees with its
stationary (infinite-past) version up to $o_{\mathrm{a.s.}}(1)$, uniformly near
$\theta_0$.
\emph{(3) Continuity.} $\theta\mapsto v_t(\theta)$ is uniformly Lipschitz, so
consistency transfers: $v_t(\hat\theta_T)-v_t(\theta_0)=o_p(1)$.
\emph{(4) Correct specification.} At $\theta_0$ the filter is exact,
$v_t(\theta_0)=\It^2$ (Definition~\ref{def:It}).
Chaining (2)--(4), $v_t(\hat\theta_T)=\It^2+o_p(1)$, and
$|\sqrt a-\sqrt b|\le\sqrt{|a-b|}$ gives $\Ihat_t^{\mathrm{ssm}}\xrightarrow{p}\It$.
As a global nonlinear functional of the series, $\Ihat_t^{\mathrm{ssm}}$ lies
outside Definition~\ref{def:local}, so the floor of
Corollary~\ref{cor:dichotomy}(ii) does not bind. The detailed proof --- the
$O(\rho^t)$ forgetting bound, the Lipschitz constant, and the offline/online
distinction --- is in Supplement~S10.
\end{proof}

\begin{remark}[Correct specification, and what is not yet proved]
\label{rem:misspec}
Proposition~\ref{prop:achieve} is a best case: $v_t(\theta_0)=\It^2$ relies on the
fitted family containing the truth, not on robustness. Under misspecification the
MLE targets a KL projection $\theta_*$ with $v_t(\theta_*)\ne\It^2$ in general, a
gap we probe against exact ground truth (Supplement~S2, $200$ paths per condition).
The estimator stays \emph{robust} where the regime structure remains identifiable
--- $t_4$ emissions under a Gaussian fit (Spearman $0.70$) and weakly separated
regimes ($0.52$) both clear the wall --- but \emph{erodes} under the wrong state
count: a three-state truth fitted by two states gives $0.35$, still ahead of the
conditional-variance estimator ($0.27$). A matched model thus degrades gracefully
under emission and separation error but needs roughly the right latent dimension.
\end{remark}

\section{Numerical Study}
\label{sec:numerical}

\subsection{Simulation: which estimator, and when}
\label{sec:sim}

We evaluate the menu against the \emph{exact} conditional $\It$ of
Section~\ref{sec:Itproperties}: closed-form for the AR(1) and AR--GARCH
processes, and by the forward filter for the HMM and the regime-switching AR.
Two hundred independent paths are drawn per process at $T=400$, and the whole
comparison is repeated across the length grid $T\in\{200,400,800,1600,3200\}$
(Table~\ref{tab:scaling}); for each path and estimator we record the Spearman
rank correlation with the true trajectory (scale-free tracking), and the per-series wall-clock cost. Point accuracy is reported by two target-blind
metrics: an \emph{independently calibrated} relative RMSE --- a single affine map
fitted on $100$ paths and scored on the disjoint $100$ (median over those held-out
paths), so it never uses
the evaluated path's own scale --- and the signed relative \emph{level bias}
$(\overline{\Ihat}-\overline{\It})/\overline{\It}$. A per-path best-affine
\emph{oracle}-calibrated RMSE, which lets each estimate borrow the target's scale
and location, isolates trajectory \emph{shape} but flatters accuracy, so we
relegate it --- with the full uncalibrated battery --- to Supplement~S9 (where it
also reproduces the Spearman ordering). The distinction is material: for the
conditional-variance estimator on AR--GARCH the oracle-calibrated error ($0.16$)
rises to $0.53$ under independent calibration, so the robust facts are the rank
tracking and the near-zero level bias, not the flattered oracle figure.
Table~\ref{tab:sim} reports medians over the $200$ paths with $95\%$ Monte-Carlo
intervals (percentile bootstrap over paths). Two design
points follow from Definition~\ref{def:It}: AR(1) is degenerate
($\It=|\phi|\sigma_\varepsilon$ is constant, so a rank correlation is undefined
and it serves only as a level check, Example~\ref{ex:ar1}); and the AR--GARCH
process is the regime in which $\It$ varies smoothly through conditional
volatility, while the HMM and regime-switching AR are the state-driven regime, in
which $\It$ jumps with the latent state. Beyond the five estimators we also run,
on the HMM, an \emph{oracle-predictor control}: the windowed estimator fed the
exact forecast revision $D_t$ from the true filter, which sets the predictor
error $\varepsilon_s\equiv 0$ and isolates the second-moment term $\delta_t$.

\begin{table}[ht]
\centering
\small
\caption{Estimator menu against the exact $\It$: median over $200$ independent
paths at $T=400$, with $95\%$ Monte-Carlo intervals for the Spearman correlation
(percentile bootstrap over paths). The principal accuracy metrics are the
scale-free Spearman rank correlation and an \emph{independently calibrated}
relative RMSE (one affine map fitted on $100$ paths and scored on the disjoint
$100$, reported as the median over those held-out paths, $\sigma_{\It}$-normalised),
with the signed relative level bias
$(\overline{\Ihat}-\overline{\It})/\overline{\It}$; cost is per series, timed
serially (hardware in Supplement~S7). The per-path oracle-calibrated RMSE, which
uses the target's own scale and location and is therefore optimistic, is reported
only as a shape diagnostic in Supplement~S9. AR(1) is a level check only ($\It$
constant, rank correlation undefined). Boldface marks the matched estimator per
process; scaling across $T\in\{200,\dots,3200\}$ is in Table~\ref{tab:scaling}.}
\label{tab:sim}
\resizebox{\textwidth}{!}{%
\begin{tabular}{llrrrr}
\toprule
Process & Estimator & Spearman [95\% CI] & RMSE (indep.) & level bias & cost (ms)\\
\midrule
\multirow{5}{*}{\shortstack[l]{AR--GARCH\\(volatility-driven)}}
 & A: windowed RMS      & $\phantom{-}0.62$ {\scriptsize$[0.60,0.64]$} & $0.80$ & $-0.37$ & $8$\\
 & B: innovation-scaled & $\phantom{-}0.81$ {\scriptsize$[0.80,0.81]$} & $0.64$ & $-0.03$ & $\mathbf{2}$\\
 & cv: AR $+$ GARCH     & $\phantom{-}\mathbf{0.98}$ {\scriptsize$[0.98,0.99]$} & $\mathbf{0.53}$ & $-0.02$ & $23$\\
 & boot: block          & $\phantom{-}0.51$ {\scriptsize$[0.50,0.53]$} & $0.89$ & $-0.03$ & $3443$\\
 & ssm: state-space     & $\phantom{-}0.91$ {\scriptsize$[0.88,0.93]$} & $0.79$ & $-0.03$ & $4789$\\
\midrule
\multirow{5}{*}{\shortstack[l]{HMM\\(state-driven)}}
 & A: windowed RMS      & $-0.01$ {\scriptsize$[-0.02,-0.00]$} & $1.00$ & $-0.35$ & $8$\\
 & B: innovation-scaled & $-0.03$ {\scriptsize$[-0.04,-0.02]$} & $1.00$ & $+0.22$ & $2$\\
 & cv: AR $+$ GARCH     & $-0.15$ {\scriptsize$[-0.19,-0.11]$} & $1.00$ & $+0.23$ & $23$\\
 & boot: block          & $-0.01$ {\scriptsize$[-0.02,-0.00]$} & $1.00$ & $+0.22$ & $3443$\\
 & ssm: state-space     & $\phantom{-}\mathbf{0.65}$ {\scriptsize$[0.64,0.67]$} & $\mathbf{0.55}$ & $-0.02$ & $4789$\\
\midrule
\multirow{5}{*}{\shortstack[l]{Regime-switching\\AR (state-driven)}}
 & A: windowed RMS      & $\phantom{-}0.14$ {\scriptsize$[0.11,0.16]$} & $1.00$ & $-0.46$ & $8$\\
 & B: innovation-scaled & $\phantom{-}0.00$ {\scriptsize$[-0.01,0.01]$} & $1.00$ & $+0.06$ & $2$\\
 & cv: AR $+$ GARCH     & $\phantom{-}0.15$ {\scriptsize$[0.13,0.16]$} & $0.99$ & $+0.07$ & $23$\\
 & boot: block          & $-0.01$ {\scriptsize$[-0.03,0.02]$} & $1.00$ & $+0.06$ & $3443$\\
 & ssm: state-space     & $\phantom{-}\mathbf{0.82}$ {\scriptsize$[0.79,0.85]$} & $\mathbf{0.62}$ & $-0.02$ & $4789$\\
\bottomrule
\end{tabular}}
\end{table}

\begin{table}[ht]
\centering
\small
\caption{Scaling of tracking accuracy with series length: median Spearman vs the
exact $\It$ over $200$ paths at each $T$ (AR(1) omitted, $\It$ constant). The
pattern is the paper's theory made empirical. \emph{Matched} models show the
\textbf{improving tracking} predicted by the model-based theory --- the
conditional-variance model on the volatility-driven
AR--GARCH and the state-space model on the state-driven processes have rank
correlation rising steadily toward one as $T$ grows --- whereas the generic non-state estimators remain poor
on the state-driven processes at all $T$. For the lag-only members, this pattern is
consistent with Corollary~\ref{cor:dichotomy}(ii); that result does not cover the
fitted conditional-variance model or the latent-state designs. The full table with
$95\%$ Monte-Carlo intervals and per-path oracle-calibrated RMSE is in
Supplement~S8; the widest
intervals occur for the small-sample state-space fits at $T=200$. Boldface marks
the matched estimator per process.}
\label{tab:scaling}
\resizebox{\textwidth}{!}{%
\begin{tabular}{llccccc}
\toprule
Process & Estimator & $T{=}200$ & $400$ & $800$ & $1600$ & $3200$\\
\midrule
\multirow{5}{*}{\shortstack[l]{AR--GARCH\\(vol.-driven)}}
 & A: windowed RMS      & $0.57$ & $0.62$ & $0.65$ & $0.65$ & $0.66$\\
 & B: innovation-scaled & $0.79$ & $0.81$ & $0.83$ & $0.84$ & $0.84$\\
 & \textbf{cv}: AR $+$ GARCH & $\mathbf{0.97}$ & $\mathbf{0.98}$ & $\mathbf{0.99}$ & $\mathbf{1.00}$ & $\mathbf{1.00}$\\
 & boot: block          & $0.44$ & $0.51$ & $0.55$ & $0.57$ & $0.59$\\
 & ssm: state-space     & $0.40$ & $0.91$ & $0.97$ & $0.98$ & $0.98$\\
\midrule
\multirow{5}{*}{\shortstack[l]{HMM\\(state-driven)}}
 & A: windowed RMS      & $-0.00$ & $-0.01$ & $-0.01$ & $-0.01$ & $-0.01$\\
 & B: innovation-scaled & $-0.01$ & $-0.03$ & $-0.02$ & $-0.03$ & $-0.02$\\
 & cv: AR $+$ GARCH     & $-0.15$ & $-0.15$ & $-0.13$ & $-0.12$ & $-0.12$\\
 & boot: block          & $-0.01$ & $-0.01$ & $-0.02$ & $-0.02$ & $-0.02$\\
 & \textbf{ssm}: state-space & $\mathbf{0.60}$ & $\mathbf{0.65}$ & $\mathbf{0.71}$ & $\mathbf{0.76}$ & $\mathbf{0.80}$\\
\midrule
\multirow{5}{*}{\shortstack[l]{Regime-sw.\\AR (state-dr.)}}
 & A: windowed RMS      & $\phantom{-}0.11$ & $0.14$ & $0.14$ & $0.15$ & $0.15$\\
 & B: innovation-scaled & $\phantom{-}0.01$ & $0.00$ & $0.01$ & $-0.00$ & $0.00$\\
 & cv: AR $+$ GARCH     & $\phantom{-}0.12$ & $0.15$ & $0.14$ & $0.15$ & $0.17$\\
 & boot: block          & $-0.02$ & $-0.01$ & $0.01$ & $-0.00$ & $0.00$\\
 & \textbf{ssm}: state-space & $\mathbf{0.60}$ & $\mathbf{0.82}$ & $\mathbf{0.94}$ & $\mathbf{0.97}$ & $\mathbf{0.99}$\\
\bottomrule
\end{tabular}}
\end{table}

\begin{figure}[t]
\centering
\includegraphics[width=0.78\linewidth]{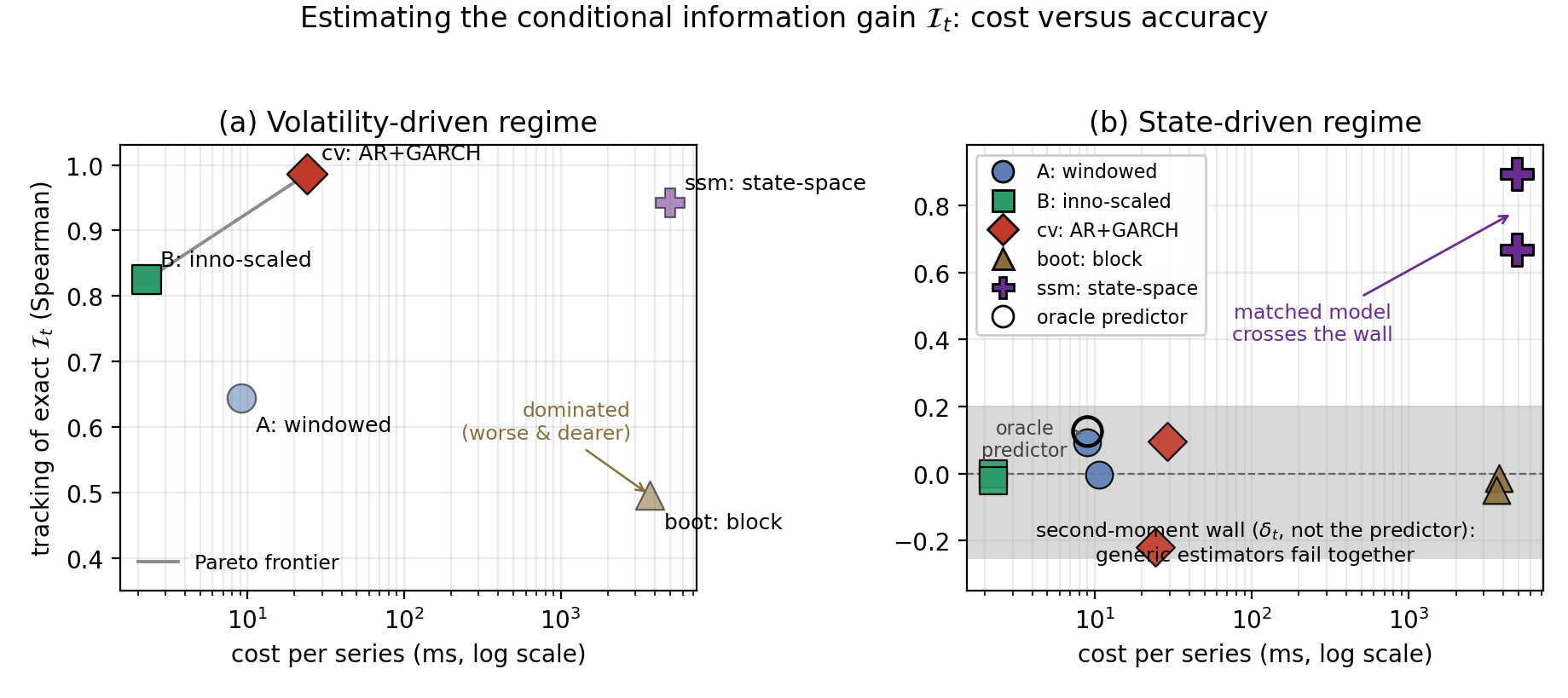}
\caption{\textbf{Central result: the cost--accuracy frontier for estimating the
conditional forecast-revision scale $\It$} (median over $200$ paths; cf.\
Table~\ref{tab:sim}). \emph{(a)} When $\It$ is volatility-driven, the
conditional-variance estimator attains the best accuracy at near-minimal cost,
and the nonparametric block bootstrap is \emph{Pareto-dominated as a point
estimator} --- less accurate \emph{and} over $100\times$ dearer, though it retains
its value for interval estimation. \emph{(b)} When $\It$ is state-driven, the
generic estimators collapse to a wall of near-zero accuracy regardless of cost;
an oracle predictor (open marker) does \emph{not} lift them, so the wall is a
second-moment-tracking limit ($\delta_t$), not the predictor term
($\norm{\varepsilon}$) of Lemma~\ref{lem:decomp}. The matched state-space
estimator (ssm) crosses it, recovering $\It$ at the cost of the dearest fit.}
\label{fig:frontier}
\end{figure}

Figure~\ref{fig:frontier} is the paper's central result, and it separates cleanly
by \emph{what drives} $\It$. On the volatility-driven AR--GARCH the
conditional-variance estimator tracks the truth almost exactly (Spearman $0.98$,
level bias $-2\%$, independently calibrated RMSE $0.53$) at a fraction of the
bootstrap's cost, sitting at the top-left of the frontier, while the block
bootstrap is Pareto-dominated \emph{as a point estimator}
(Figure~\ref{fig:frontier}a). The streaming surrogates follow ($\text{B}=0.81$,
$\text{A}=0.62$); the rolling block bootstrap tracks \emph{worst} ($0.51$) at over
a hundredfold the cost, because on a single path it resamples a trailing window and
so estimates a local \emph{average} second moment, smoothing exactly the variation
the others follow. Its value is the uncertainty band and the replicated regime
(Theorem~\ref{thm:bbconsistency}), not single-path point accuracy (Supplement~S5).

On the HMM and regime-switching processes \emph{every generic estimator} is
near-zero-correlated with $\It$ (Table~\ref{tab:sim}). The natural reading ---
predictor misspecification, since these use a linear predictor on a nonlinear
mean --- is ruled out by the oracle-predictor control: feeding the windowed
estimator the \emph{exact} revision ($\varepsilon_s\equiv0$) still yields only
Spearman $0.13$, $0.08$, $0.06$ at $w=10,25,50$ (median over $200$ paths). Even an
exact predictor barely helps, so the binding constraint is $\delta_t$ --- a
windowed second moment cannot track a moment that jumps with the latent state.

The remedy is a better \emph{model}, not a better predictor. The matched
state-space estimator (Section~\ref{sec:ssm}) tracks $\It$ at Spearman $0.65$ (HMM)
and $0.82$ (RS-AR), crossing the wall the generic estimators and the oracle
predictor hit (Figure~\ref{fig:frontier}b). This empirical wall is \emph{consistent
with} the lag-only bound of Theorem~\ref{thm:floor} but not identical to it (there
the switching state is observed; here it is latent and the failure is documented,
not proved for every smoother). The matched fit is the dearest entry and pays only
once the cheaper models fail: on AR--GARCH it stays accurate ($0.91$) but is
Pareto-dominated by the conditional-variance model at $\sim\!200\times$ less cost.
The scaling grid (Table~\ref{tab:scaling}) matches the theory: as $T$ grows the
matched estimators' tracking improves toward one (cv to $1.00$; ssm to $0.80$ and
$0.99$ at $T=3200$) --- the empirical signature of the pointwise consistency of
Proposition~\ref{prop:achieve} (established for $\Ihat_t^{\mathrm{ssm}}$; the
analogous GARCH-plug-in argument for $\Ihat_t^{\mathrm{cv}}$ is standard but not
carried out here) --- while the generic estimators stay pinned at the wall for
\emph{all} $T$.
The lesson is not that one estimator dominates but that one must match the model to
the second-moment structure of $\It$ --- which Section~\ref{sec:selection}
operationalises.

\paragraph{Scope and robustness}
Every quantitative claim here is conditional on the synthetic designs, where exact
ground truth is available: the numbers are evidence for the \emph{mechanisms}, not
universal constants, and the real-data checks (Section~\ref{sec:realdata}) are
illustrations needing reconfirmation on further series. The conclusions recur
\emph{across} designs rather than at a single point --- three process families,
five lengths $T\in\{200,\dots,3200\}$ ($200$ paths per cell with $95\%$ intervals,
Tables~\ref{tab:sim}--\ref{tab:scaling}), three recurrent spectral radii in the
gate study, and three forms of misspecification (Supplement~S2) --- with
separations exceeding Monte-Carlo error. We state the conclusions conservatively:
the floor is proved only for the specified estimator class, achievability only
under correct specification, and the cost inversion only for the volatility-driven
designs studied. Reproducibility details and the machine manifest are in the
Code and Data Availability statement and Supplement~S7.

\subsection{A Candidate Free Estimator: The Forget Gate}
\label{sec:S4}

The remaining candidate would be free, which is why it is worth ruling out
carefully. A network trained for one-step prediction already computes retention
gates $f_t\in(0,1)$ at every step. Two readings must be kept apart. The
\emph{standard} one, which the LSTM and GRU literature supports, is that these
gates regulate how much of the past to retain versus how much a new observation
influences the state
\citep{hochreiter1997long,gers2000learning,cho2014learning}. A \emph{stronger}
statistical claim --- that a gate value or gate ratio \emph{estimates} a particular
conditional moment --- does not follow from that interpretation. In our own earlier
work on this program we advanced exactly the stronger claim, as a \emph{heuristic}
rather than a theorem: the relation $f_t\approx\It/\mathcal I_{t-1}$, proposed but
not established and not tested against exact history-specific ground truth. If it
held, $\It$ would come for nothing wherever such a network is already trained. We
therefore treat that proxy as an \emph{empirical hypothesis} and test it directly
against the exact conditional $\It$ --- LSTM \citep{hochreiter1997long} and GRU
\citep{cho2014learning} cells, three processes, three spectral radii, $18$
configurations.

The \emph{free} reading does not survive, but the honest answer is two-sided. The
\emph{averaged} retention gate does \emph{not} track $\It$ (pooled Spearman
$\approx 0$ in both the ratio form $\bar f_t\!\sim\!\It/\mathcal I_{t-1}$ and a
level form, with the networks well trained, one-step
$\mathrm{MSE}/\mathrm{Var}\approx0.57$--$0.64$), so the naive proxy is refuted. Yet
the information is not \emph{absent}: a \emph{trained linear probe} on the full
$32$-unit gate vector, fitted on half the sequences and scored out of sample on the
other half, recovers $\It$ well (Spearman $0.86$--$0.94$, $R^2$ up to $0.89$), far
above the averaged gate and in the range of the matched state-space estimator.
Individual units carry dynamic range that cancels in the mean, but a linear
combination reconstructs the signal. The representation therefore \emph{does}
encode $\It$; what fails is reading it off for free.

The practical lesson is exact: as a drop-in, label-free estimator the forget gate
should \emph{not} be used (Section~\ref{sec:selection}), since the naive read-outs
do not recover $\It$ and the probe that does is \emph{supervised}, needing the very
quantity one lacks. The probe's success is nonetheless a positive signal that a
decoder trained on simulated ground truth, or an architecture designed to expose
$\It$, could make the representation usable --- a downstream question left to future
work. Provenance, the training and early-stopping protocol, the linear-probe
control, and the per-unit permutation analysis are in Supplement~S1.

\subsection{Real-Data Illustrations}
\label{sec:realdata}

The synthetic studies supply exact ground truth; the real series do not, so
estimator agreement must \emph{not} be read as ground truth. For each series we
report the $\It$ trajectory each estimator produces, pairwise agreement,
per-series cost, residual diagnostics, and an external predictive check ---
whether a larger estimated $\It$ forecasts a larger realised improvement from
updating the one-step mean rather than retaining the pre-update forecast. Two
series contrast the sources of variation: daily S\&P~500 log-returns
(2004--2023, $n=5032$), plausibly volatility-driven, and annual sunspot counts
($n=309$, square-root scaled), a classical nonlinear-time-series benchmark
\citep{tong1990nonlinear}.

\paragraph{Volatility example: agreement as a cost screen}
\label{sec:realdata-fin}

On the S\&P~500 returns the estimators have high pairwise rank agreement
($0.79$--$0.97$; conditional-variance and innovation-scaled at $0.92$). This is
\emph{compatible with} a volatility-driven interpretation, but agreement may also
arise from shared smoothing or common misspecification, so it is not proof. The
conditional-variance and streaming estimators should therefore be compared by
rolling predictive loss, interval calibration, and standardized-residual checks;
among candidates that pass, cost can decide, and the spread is stark --- the
innovation-scaled surrogate runs in $32$\,ms against the rolling block bootstrap's
$\mathbf{56}$\,\textbf{seconds}. The bootstrap remains useful when an interval
around the local second-moment estimate is required: agreement alone does not
imply that its uncertainty calculation ``buys nothing.''

\emph{An external, prequential predictive check.} The risk identity of
Section~\ref{sec:Itproperties} makes the estimate testable without ground truth.
For estimated forecasts $\hat m_s$ with error $\varepsilon_s=\hat m_s-m_s$, the law
of total variance gives the \emph{exact} relation
\[
  \E[\Delta_t\mid\F_{t-1}]
  =\It^2+\varepsilon_{t-1}^2-\E[\varepsilon_t^2\mid\F_{t-1}],
  \qquad \Delta_t=(X_{t+1}-\hat m_{t-1})^2-(X_{t+1}-\hat m_t)^2 .
\]
Under \emph{oracle} forecasts ($\varepsilon_s\equiv0$) this collapses to $\It^2$,
so unit slope in a regression of $\Delta_t$ on $\Ihat_t^2$ is the calibration
benchmark only in that case; with estimated forecasts the predictor-error terms
remain, and we therefore read the regression as a \emph{directional} diagnostic
rather than a calibration test. We run this \emph{strictly prequentially} --- the AR(1)
mean and both estimators refitted on an expanding window (monthly), so every
forecast at $t$ uses only data up to $t$. The slope is positive for both ---
$1.60$ (Newey--West SE $0.71$) for Surrogate~B and $2.40$ ($1.27$) for
$\Ihat_t^{\mathrm{cv}}$ --- and the highest-$\Ihat_t$ decile shows a clearly
positive mean reduction ($\approx 0.22$) against a near-zero bulk (a look-ahead
full-sample fit inflates the slopes to $2.8$ and $3.8$). We read this at the
strength the evidence supports: the \emph{direction} is confirmed, and the slopes
above one are \emph{consistent with} a lag-only estimator under-stating $\It^2$ at
volatility spikes (the under-smoothing of Theorem~\ref{thm:floor}), but the HAC
intervals are wide and smoothing lag need not be the sole cause, so this is a
directional confirmation, not a calibrated slope.
A full rolling comparison across the menu, with calibration curves, is left to
future work.

\paragraph{Sunspots: disagreement as a model-checking signal}
\label{sec:realdata-nl}

On the sunspot series the estimators \emph{diverge}: the two windowed surrogates
agree with each other ($0.84$) while the conditional-variance estimator parts
company (agreement $0.08$ and $0.24$), finding little volatility signal and
returning a near-flat estimate. That disagreement \emph{rejects} the claim that
all candidates estimate the same trajectory, but it does not by itself identify a
state-driven mechanism or crown a state-space estimator. The next step is to
compare nonlinear-mean, conditional-variance, and state-space models by
prequential predictive loss and stability across model order, initialisation, and
training period; a state-space estimate is reported only as a candidate
trajectory, with uncertainty and sensitivity analysis. The conclusion is
therefore diagnostic: the simple volatility interpretation is inadequate for this
series, and any additional structure must be validated out of sample, which is
what the selection workflow keys on (Section~\ref{sec:selection}). We present the
sunspot series as a brief \emph{illustration} of this diagnostic step, not as a
validated state-driven conclusion; the substantive real-data evidence is the
prequential S\&P~500 check above.

\subsection{A Diagnostic Selection Workflow}
\label{sec:selection}

The evidence supports a \emph{provisional workflow}, not a universally validated
decision rule. It qualifies the reflex of reaching for the most expensive
estimator and keys the choice to \emph{what generates} $\It$; but estimator
agreement is a low-cost warning signal, not a proof of correctness, and the
workflow should be calibrated in simulation by its sensitivity and specificity
for distinguishing smooth volatility, persistent state changes, nonlinear means,
heavy tails, and long memory. Algorithm~\ref{alg:selection} states the decision
compactly; the paragraphs below elaborate each step and its thresholds.

\begin{algorithm}[ht]
\caption{Choosing an estimator of $\It$ (diagnostic selection workflow).}
\label{alg:selection}
\begin{algorithmic}[1]
\Require series $x_{1:T}$; target mode (offline trajectory / online / interval); budget
\State \textbf{Fix target and mode.} Do not read an offline full-sample fit as an
  online claim without a rolling/prequential refit.
\State \textbf{Diagnose the structure} (free): compute the cheap estimators
  A, B, cv and set $\rho_{\mathrm{pair}}=\min$ of their three pairwise Spearman
  correlations (a conservative agreement summary).
\If{$\rho_{\mathrm{pair}} > 0.80$ \textbf{and} the estimates co-move with the fitted GARCH conditional s.d.\ (Spearman ${>}0.7$)}
  \State \textbf{Volatility-driven} $\Rightarrow$ use $\Ihat_t^{\mathrm{cv}}$
    (GARCH-type, Sec.~\ref{sec:cvestimator}); under an $O(1)$/streaming budget use Surrogate~B.
\ElsIf{$\rho_{\mathrm{pair}} < 0.30$ \textbf{and} the generic estimates are each within Monte-Carlo error of a constant ($|$Spearman with the fitted volatility$|{<}0.2$)}
  \State \textbf{State-driven / nonlinear} $\Rightarrow$ fit a matched latent-state
    model $\Ihat_t^{\mathrm{ssm}}$ (Sec.~\ref{sec:ssm}); pick model order by
    \emph{out-of-sample} predictive loss, not by agreement with the unobserved $\It$.
\Else
  \State \textbf{Ambiguous} $\Rightarrow$ compare nonlinear-mean, conditional-variance
    and state-space models by prequential loss; report order/initialisation sensitivity.
\EndIf
\If{an interval is required}
  \State add the block bootstrap for a band around the local estimate, under-smoothing
    per Thm.~\ref{thm:bbconsistency}(ii); do \emph{not} use it as a single-path point estimator.
\EndIf
\State \textbf{Do not} substitute a trained-network forget gate for $\Ihat_t$ (Sec.~\ref{sec:S4}).
\end{algorithmic}
\end{algorithm}

\textbf{(i) Check consistency with a volatility-driven model.} 
Estimator \emph{agreement} is a free diagnostic. When the generic estimators agree and
track, the volatility-driven reading is not contradicted; when they \emph{disagree}
and all collapse toward zero --- as on the sunspot series
(Section~\ref{sec:realdata}) and the state-driven synthetic processes --- the
data are \emph{inconsistent with} the simple volatility-driven model, which
motivates comparison with nonlinear-mean, latent-state, and other
conditional-moment models rather than proving any one mechanism. On the synthetic
processes, where the mechanism is known, the oracle-predictor control
(Section~\ref{sec:sim}: even an exact predictor leaves Spearman $\le 0.13$)
further shows the remedy is a matched model, not a cheaper or better-predictor
one. Agreement \emph{strengthens} but does not \emph{establish}
the volatility reading.

\textbf{(ii) Choose by regime and budget.} 
For \emph{volatility-driven} $\It$, use
$\Ihat_t^{\mathrm{cv}}$ --- Spearman $0.98$ on AR--GARCH at a few ms per series ---
or Surrogate~B under a streaming/$O(1)$ budget; use Surrogate~A when even a
first-order mean model is unavailable. For \emph{state-driven} $\It$, use
the matched state-space estimator $\Ihat_t^{\mathrm{ssm}}$, the only entry that
tracks it ($0.65$ and $0.82$, rising toward one with $T$); select among candidate
models by \emph{out-of-sample} predictive loss, not agreement with the unobserved
$\It$, and report order and initialisation sensitivity --- it is robust to emission
and separation error but erodes under the wrong state count (Supplement~S2), and is
Pareto-dominated by $\Ihat_t^{\mathrm{cv}}$ on volatility-driven $\It$ at
$\sim\!200\times$ less cost. When \emph{uncertainty bands} are wanted, use the block
bootstrap (under-smoothing per Theorem~\ref{thm:bbconsistency}(ii)); it is not a
single-path point-estimate default.

\textbf{(iii) What to avoid.} 
Use the block bootstrap for \emph{intervals}, not
single-path \emph{point} estimates (there it is Pareto-dominated by
$\Ihat_t^{\mathrm{cv}}$, $>100\times$ dearer and less accurate); once the cheap
estimators agree and track, the dearer ones are redundant unless the diagnosis is
ambiguous or a band is required; and do not substitute a trained forget gate for
$\Ihat_t$ (Section~\ref{sec:S4}).

\subsection{Discussion}
\label{sec:discussion}

Locating each of the paper's three qualifications precisely --- each specific to
this target, these estimator classes, and the designs studied, not a challenge to
the broader validity of bootstrap, forecasting, or gated-network methods --- is
what turns it into guidance.

\emph{Additional resampling need not improve point tracking.} The block bootstrap
only re-estimates the same smoothed local moment (Theorem~\ref{thm:bbconsistency})
\citep{kunsch1989jackknife,politis1994stationary,lahiri2003resampling}, so on the
volatility-driven designs it is Pareto-dominated \emph{as a point estimator} by a
matched conditional-variance model; reserve resampling for the intervals it is
built for.

\emph{A better predictor need not help.} The oracle-predictor control rules out
predictor misspecification as the cause of the state-driven failure --- an exact
forecast still leaves Spearman $\le 0.13$ (Section~\ref{sec:sim}) --- so the binding
constraint is second-moment tracking (Lemma~\ref{lem:decomp},
Theorem~\ref{thm:floor}), and the remedy is a model matched to the mechanism, which
recovers $\It$ (Proposition~\ref{prop:achieve}) and converges on it as $T$ grows.

\emph{A trained gate encodes $\It$, but not for free.} Read through the
value-of-information lens
\citep{hochreiter1997long,gers2000learning,cho2014learning,bialek2001predictability,itti2009bayesian},
a forget gate looks like a free estimate of $\It$; but against exact ground truth
(Section~\ref{sec:S4}) the naive read-outs do not track $\It$, whereas a supervised
probe on the full gate vector recovers it. The representation is rich enough to
expose $\It$ if decoded deliberately --- just not where a label-free read-out
finds it.

\paragraph{What this leaves for later} 
The theoretical scope is deliberately tight:
Theorem~\ref{thm:floor} is a sharp impossibility for externally tuned lag-only
smoothers under a predictable sampling-scale switch and
Proposition~\ref{prop:achieve} a matched-model achievability result, but we do
\emph{not} prove a single minimax lower-and-upper theorem for one latent-state
model, and regard unifying the two constructions as the main open problem. The
practical extensions are a model that identifies the second-moment structure
\emph{automatically} --- a learned state-space or mixture-density estimator, or an
adaptively selected state count --- and a multivariate $\It$ through the
conditional covariance of the vector revision; architectures built to expose $\It$
directly are a separate, downstream question.

\section{Conclusion}
\label{sec:conclusion}

This paper asked which estimator of the conditional forecast-revision scale $\It$
to use under a given computational budget, and answered it against exact ground
truth. A single error decomposition (Lemma~\ref{lem:decomp}) separates predictor
error from second-moment tracking and organizes the menu. Two lessons emerge. When
$\It$ is volatility-driven, the usual cost ordering inverts: in the designs studied
a cheap conditional-variance model is more accurate \emph{and} far cheaper than the
block bootstrap, whose single-path value lies in the interval it supplies, not in a
point estimate. When $\It$ is state-driven, the generic non-state estimators fail
together --- an oracle-predictor control shows the binding term is second-moment
tracking, not the predictor --- while a correctly specified state-space model
recovers $\It$ where they read zero. The practical rule is to match the estimator's
model to the structure generating $\It$ and to pay for the dear fit only when the
cheaper models prove inadequate (Section~\ref{sec:selection}). The forget gate,
proposed elsewhere as a free proxy, does not survive as a label-free estimator,
though a supervised probe recovers $\It$ from the full gate population --- a precise
cautionary negative rather than a usable shortcut. The scope is deliberately
narrow --- the smoothing limit is proved for one estimator class and achievability
only under correct specification (open problems in Section~\ref{sec:discussion})
--- but within it the message is sharp: match the estimator to the second-moment
structure of $\It$: structure decides which estimators are adequate, and cost then
chooses among them.

\section*{Declaration of Competing Interest}

The authors declare that they have no known competing financial interests or
personal relationships that could have appeared to influence the work reported in
this paper.

\section*{Funding}

This work was supported by the National Science and Technology Council (NSTC),
Taiwan, under grant~115-2118-M-001-004-.

\section*{Code and Data Availability}

The two real series are public: daily S\&P~500 closing prices (Yahoo Finance
\texttt{\textasciicircum GSPC}, 2004--2024, analysed as percentage log-returns) and
the annual sunspot numbers (\texttt{statsmodels}/WDC-SILSO, square-root scaled). All
code, random seeds, the synthetic series, the trained network weights (or the seeds
that regenerate them), the software and machine manifest, the cached data with
retrieval dates and preprocessing, and a figure/table-to-script mapping will be
released under permissive licences at
\url{https://github.com/knight-ivan/Conditional-Forecast-Revision} upon acceptance.

\appendix
\section{Proof of the Central Bias--Variance Identity}
\label{app:proofs}

The proof of Theorem~\ref{thm:bbconsistency} (block-bootstrap consistency and
validity), the bias decompositions for the lightweight surrogates, and the
complete proof of Proposition~\ref{prop:achieve} (matched-model achievability) are
given in the Supplementary Material (Sections~S3, S4, and~S10). We give here the
short proof of the central limitation result.

\subsection{Proof of Theorem~\ref{thm:floor} and Corollary~\ref{cor:dichotomy}}
\label{app:floor}
Condition on the variance path $\{v_r\}$ and the tuning field $\mathcal T$. With
$\eta_s=D_s^2-v_s$, we have
$\hat v_t-v_t = B_t + \sum_{s\in W_t}a_{t,s}\eta_s$ where
$B_t=\sum_{s\in W_t}a_{t,s}(v_s-v_t)$.
The conditional mean of the second term is zero, and conditional uncorrelatedness
gives
$\E[(\hat v_t-v_t)^2\mid\{v_r\},\mathcal T]
 = B_t^2 + \sum_{s\in W_t}a_{t,s}^2\,\Var(D_s^2\mid\F_{s-1})$.
The two bounds follow from $v_{\min}\le v_s\le v_{\max}$ and
$\sum_s a_{t,s}^2=1/w_t^{\mathrm{eff}}$. The step $\E[\eta_s\mid\{v_r\},\mathcal T]=0$ uses
exogeneity of the variance path; it would fail for an endogenous path that is a
function of past $\eta_s$ (e.g.\ GARCH), and likewise for weights constructed from
the same $\eta_s$, which is why $\{v_r\}$ is taken exogenous and
Definition~\ref{def:local} excludes data-dependent weights.

\emph{Corollary (i).} Uniform weights give
$|B_t|\le\tfrac1w\sum_{s\in W_t}L|s-t|/T=O(w/T)$ and
$\sum_s a_{t,s}^2=1/w$, so Theorem~\ref{thm:floor} yields $O(w^2/T^2+w^{-1})$,
minimised in order at $w\asymp T^{2/3}$. Since $v_t\ge v_{\min}>0$,
$|\sqrt{\hat v_t}-\sqrt{v_t}|=|\hat v_t-v_t|/(\sqrt{\hat v_t}+\sqrt{v_t})$, and
localisation to a neighbourhood of $v_t$ gives the stated delta-method rate.

\emph{Corollary (ii).} Substitute $v_s=\mu+\delta\zeta_{s-1}$:
$B_t=\delta\bigl(\sum_{s\in W_t}a_{t,s}\zeta_{s-1}-\zeta_{t-1}\bigr)$. The weighted
term is $\F_{t-2}\vee\mathcal T$-measurable (the weights are $\mathcal T$-measurable
and $s\le t-1$), whereas the target depends on
$\zeta_{t-1}$; since $\mathcal T$ is independent of the evaluation path, the
minimum-mean-squared-error property of conditional
expectation gives
$\E(B_t^2)\ge\delta^2\E\{\Var(\zeta_{t-1}\mid\F_{t-2})\}=\delta^2(1-\rho^2)$, the
equality from the Markov property and $\E(\zeta_{t-1}\mid\zeta_{t-2})=\rho\zeta_{t-2}$.
Adding the nonnegative innovation-variance term proves the floor. The one-step
index shift --- $v_t=\mu+\delta\zeta_{t-1}$ is $\F_{t-1}$-measurable --- is
essential for compatibility with the definition of $\It$. $\qed$

\clearpage
\setcounter{section}{0}
\setcounter{subsection}{0}
\setcounter{table}{0}
\setcounter{figure}{0}
\setcounter{equation}{0}
\renewcommand{\thesection}{S\arabic{section}}
\renewcommand{\thesubsection}{S\arabic{section}.\arabic{subsection}}
\renewcommand{\thetable}{S\arabic{table}}
\renewcommand{\thefigure}{S\arabic{figure}}
\renewcommand{\theequation}{S\arabic{equation}}
\numberwithin{equation}{section}
\definecolor{navyblue}{RGB}{26,26,78}
\definecolor{royalblue}{RGB}{44,95,165}

\begin{center}
  {\Large\bfseries Supplementary Material}\\[8pt]
  {\large\bfseries Estimating the Conditional Forecast-Revision Scale in
   Sequential Models:\\[2pt]
   Local-Smoothing Limits, Matched Models, and Cost--Accuracy Trade-offs}\\[10pt]
  {Hui Mean Foo \qquad Yuan-chin Ivan Chang}\\[3pt]
  {\itshape Institute of Statistical Science, Academia Sinica, Taipei, Taiwan}
\end{center}
\medskip

\noindent This supplement collects the proofs, the forget-gate study,
misspecification robustness, and reproducibility details. Section and equation
numbers with an ``S'' prefix refer to this supplement; unprefixed references
(e.g.\ Theorem~\ref{thm:floor}) are to the main paper.
\medskip

\section{The Forget-Gate Study}
\label{supp:gate}

At the free end of the estimator menu sits a tempting idea: a network trained for
one-step prediction already computes retention gates $f_t\in(0,1)$ at no extra
cost, so if $f_t$ encoded $\It$ it would be the cheapest estimator available. This
section tests that idea against exact ground truth and finds it wanting; the main
paper (Section~\ref{sec:S4}) records only the conclusion.

\textbf{Provenance.} The relationship $f_t \approx \It/\mathcal{I}_{t-1}$ was
advanced in our own earlier work on this program as a \emph{heuristic}, not a
theorem, and was neither established there nor tested against exact
history-specific ground truth. We therefore treat it as an empirical hypothesis to
be tested against the exact conditional $\It$, not as an inherited result.

\subsection{Derivation of the forget-gate proxy}
Under the hypothesis $f_{t,j} \approx \It/\mathcal{I}_{t-1}$, telescoping gives
$\It \approx \mathcal{I}_1 \prod_{s=2}^t f_{s,j}$, suggesting a recalibrated proxy
$\Ihat_{t,\mathrm{fg}} = \hat{a}\cdot\prod_{s=2}^t f_{s,j} + \hat{b}$ with
$(\hat{a},\hat{b})$ estimated by least squares on a small labelled set. The
cumulative product is numerically fragile for long $t$ (gates in $(0,1)$ drive it
to zero), so a windowed or log-domain version is used in practice.

\subsection{Gated-architecture comparison}
We train LSTM \citep{hochreiter1997long,gers2000learning} and GRU
\citep{cho2014learning} cells
only, across the three process types and $\rho(W_h)\in\{0.3,0.6,0.9\}$,
correlating the retention gate against the \emph{exact} conditional $\It$
(Gauss--Hermite ground truth, not a bootstrap target). We restrict to the LSTM and
GRU because the question here is purely one of estimation --- can a gate trained
for prediction be \emph{read} as an estimate of $\It$? Architectures whose
internal parameters are \emph{designed} to be input-selective raise a distinct
mechanistic question that an estimation paper is not the place to settle
(main-paper discussion).

\subsection{Findings: the mean-gate proxy is not supported}
\label{supp:gatefindings}
\textbf{Headline (18 configurations, exact ground truth).} The mean retention gate
$\bar f_t$ does \emph{not} track $\It$. Across the tracking processes (HMM, RS-AR)
the pooled Spearman correlation is $\approx 0$ for both the ratio form
$\bar f_t\!\sim\!\It/\mathcal I_{t-1}$ (median $|\rho|\!<\!0.05$) and the level form
$\bar f_t\!\sim\! g(\It)$ ($|\rho|\!<\!0.14$, level-proxy $R^2\!<\!0.15$;
Table~\ref{tab:s4}). Well-trained models are used: one-step
$\mathrm{MSE}/\mathrm{Var}\!\approx\!0.57$--$0.64$, so this is not an under-fitting
artifact. The mean-gate scan (Table~\ref{tab:s4}) spans all $18$ configurations
(three processes $\times$ two architectures $\times$ three spectral radii
$\rho(W_h)\in\{0.3,0.6,0.9\}$); the per-unit and linear-probe analyses below are run
at \emph{one representative radius} per state-driven (process, architecture) cell
--- HMM at $\rho(W_h)=0.6$, and RS-AR with LSTM at $0.6$ and GRU at $0.9$ --- i.e.\
the four cells of Tables~\ref{tab:s4unit} and~\ref{tab:s4probe}, not a pooling over
$\rho(W_h)$.

\textbf{Why, and the honest nuance.} The failure is partly an \emph{averaging}
effect, not gate inertness: individual units carry real dynamic range (per-unit
gate sd $\approx 0.14$--$0.34$) that cancels in the mean ($\bar f_t$ sd
$\approx 0.01$--$0.04$, against $\It$ sd $\approx 0.11$--$0.14$). A per-unit
analysis with a circular-shift permutation null (max over all $H=32$ units) asks
whether the \emph{best} unit tracks $\It$. Effect sizes must be read here, not
$p$-values: with $\sim\!5\times10^4$ pooled points the null is so tight (95th
percentile $|\rho|\approx0.01$--$0.02$) that every configuration is
``significant'' at a trivial correlation. Only \textbf{one of four}
configurations (HMM + LSTM) shows a single unit with a materially large
correlation ($|\rho|\!=\!0.48$; Table~\ref{tab:s4unit}); the others peak at
$0.03$--$0.11$. So the mechanism \emph{can} occur --- a lone LSTM unit does encode
conditional forecast uncertainty on the HMM --- but it is neither reliable across
architectures/processes nor recoverable from the averaged gate.

\begin{table}[t]
\centering
\small
\caption{The mean forget gate does not track $\It$. Pooled Spearman correlation
between the averaged retention gate and the \emph{exact} $\It$, in the ratio form
($\bar f_t\!\sim\!\It/\mathcal I_{t-1}$) and level form ($\bar f_t\!\sim\!\It$),
pooled over all test sequences and time points. The level-proxy $R^2$ bound is
reported in the preceding text. One-step prediction error
was $\mathrm{MSE}/\mathrm{Var}\approx 0.57$--$0.64$ throughout. AR(1) is omitted
($\It$ constant, correlation undefined).}
\label{tab:s4}
\begin{tabular}{llccc}
\toprule
Process & Arch & $\rho(W_h)$ & Spearman (ratio) & Spearman (level) \\
\midrule
\multirow{6}{*}{HMM}
 & LSTM & 0.3 & $\phantom{-}0.02$ & $\phantom{-}0.06$\\
 & LSTM & 0.6 & $\phantom{-}0.02$ & $\phantom{-}0.07$\\
 & LSTM & 0.9 & $\phantom{-}0.03$ & $\phantom{-}0.09$\\
 & GRU  & 0.3 & $-0.03$ & $-0.12$\\
 & GRU  & 0.6 & $-0.02$ & $-0.11$\\
 & GRU  & 0.9 & $-0.04$ & $-0.13$\\
\midrule
\multirow{6}{*}{RS-AR}
 & LSTM & 0.3 & $-0.01$ & $-0.01$\\
 & LSTM & 0.6 & $-0.01$ & $\phantom{-}0.01$\\
 & LSTM & 0.9 & $\phantom{-}0.00$ & $\phantom{-}0.01$\\
 & GRU  & 0.3 & $\phantom{-}0.00$ & $\phantom{-}0.02$\\
 & GRU  & 0.6 & $\phantom{-}0.00$ & $\phantom{-}0.01$\\
 & GRU  & 0.9 & $\phantom{-}0.07$ & $\phantom{-}0.22$\\
\bottomrule
\end{tabular}
\end{table}

\begin{table}[t]
\centering
\small
\caption{Per-unit analysis: does the \emph{best} hidden unit track $\It$?
Individual units carry real dynamic range (per-unit gate sd) that cancels in the
mean (averaged-gate sd), against the target's sd. ``$\max|\rho|$'' is the largest
over $H=32$ units, with the circular-shift permutation-null 95th percentile in
parentheses ($2000$ permutations). Every row is ``significant'' at the tight null
($p<0.001$), but only HMM+LSTM has a \emph{materially} large single unit.}
\label{tab:s4unit}
\begin{tabular}{llccc}
\toprule
Process & Arch & gate sd (unit / mean) & $\It$ sd & $\max|\rho|$ (null $p_{95}$) \\
\midrule
HMM   & LSTM & $0.20 \,/\, 0.02$ & $0.11$ & $\mathbf{0.48}$ \;\; ($0.01$)\\
HMM   & GRU  & $0.17 \,/\, 0.02$ & $0.11$ & $0.06$ \;\; ($0.01$)\\
RS-AR & LSTM & $0.14 \,/\, 0.04$ & $0.14$ & $0.03$ \;\; ($0.02$)\\
RS-AR & GRU  & $0.14 \,/\, 0.01$ & $0.14$ & $0.11$ \;\; ($0.02$)\\
\bottomrule
\end{tabular}
\end{table}

\subsection{The strongest control: a trained linear probe on the full gate vector}
\label{supp:probe}
The mean-gate result (Table~\ref{tab:s4}) and the per-unit result
(Table~\ref{tab:s4unit}) leave one objection open: perhaps $\It$ is carried by a
\emph{linear combination} of units that neither the mean nor any single unit
exposes. We test this in its strongest form. For each state-driven configuration we
fit an ordinary-least-squares \emph{linear probe}
$\It\approx\beta_0+\beta^\top f_t$ from the full $H=32$-dimensional retention-gate
vector to the exact $\It$ on one half of the test sequences, and score it on the
\emph{disjoint} other half, so a high score cannot be overfitting ($\sim\!27{,}000$
pooled training points, $33$ parameters). The probe recovers $\It$ well out of
sample (Table~\ref{tab:s4probe}): Spearman $0.86$--$0.94$ and $R^2$ up to $0.89$,
far above the averaged-gate proxy ($\approx0$ on the same split) and in the range
of the matched state-space estimator of the main paper.

The two-sided reading is the honest one. $\It$ \emph{is} linearly
present in the trained gate population --- so predictive training does encode $\It$
in the internal state --- but the naive, label-free read-outs (the averaged gate
and the $f_t\approx\It/\mathcal I_{t-1}$ ratio) do not recover it, and the probe
that does is \emph{supervised}: it must be fitted against a known $\It$, exactly the
quantity unavailable in application. The forget gate is therefore not a free
estimator, while the probe's success is a positive signal for a supervised decoder
trained on simulated ground truth, or an architecture built to expose $\It$
directly (left to future work).

\begin{table}[t]
\centering
\small
\caption{Linear-probe control. Out-of-sample tracking of the exact $\It$ by an OLS
probe on the full $32$-unit gate vector (fitted on half the sequences, scored on
the disjoint half), against the averaged-gate proxy on the same split. The gate
population linearly encodes $\It$; the mean does not.}
\label{tab:s4probe}
\resizebox{\textwidth}{!}{%
\begin{tabular}{llccc}
\toprule
Process & Arch & probe Spearman (OOS) & probe $R^2$ (OOS) & mean-gate Spearman \\
\midrule
HMM   & LSTM & $\mathbf{0.94}$ & $0.47$ & $\phantom{-}0.06$\\
HMM   & GRU  & $\mathbf{0.93}$ & $0.57$ & $-0.12$\\
RS-AR & LSTM & $\mathbf{0.88}$ & $0.89$ & $\phantom{-}0.02$\\
RS-AR & GRU  & $\mathbf{0.86}$ & $0.86$ & $\phantom{-}0.20$\\
\bottomrule
\end{tabular}}
\end{table}

\subsection{Network training protocol}
\label{supp:training}
So that the negative result cannot be blamed on undertraining, we state the
protocol explicitly. Each cell is a single-layer LSTM or GRU with $H=32$ hidden
units and a linear read-out, trained to minimise one-step squared prediction error
by Adam (learning rate $10^{-3}$, batch size $128$) for $30$ epochs on $800$
training sequences of length $150$; the recurrent weight is rescaled after every
optimiser step so that its spectral radius equals the grid value $\rho(W_h)$
exactly. Gates are read from an explicit per-timestep recurrence (not a fused
kernel), so the analysed gate is the one actually used for prediction. The
resulting one-step $\mathrm{MSE}/\mathrm{Var}\approx0.57$--$0.64$ confirms the
networks learned genuine predictive structure (a value near $1$ would indicate no
better than predicting the mean); the gate correlations are therefore read from
well-trained, not under-fitted, models. Early stopping was not needed at this width
and horizon --- training and held-out prediction error move together and plateau
--- and all seeds, the spectral-radius schedule and the full configuration are
pinned in the released run manifest.

\section{Misspecification Robustness of the Matched State-Space Estimator}
\label{supp:misspec}

Proposition~\ref{prop:achieve} is a correct-specification result. To bound its
scope we run the fitted two-state Gaussian state-space estimator against the exact
$\It$ of processes \emph{outside} its family (Table~\ref{tab:misspec}). The
Student-$t$ ground truth is computed by grid quadrature over the predictive
$t$-mixture and validated against Monte Carlo (median relative error $0.24\%$).
The estimator is robust to the misspecifications that leave the regime structure
identifiable --- heavy-tailed ($t_4$) emissions fitted by a Gaussian model
(Spearman $0.70$) and weakly separated regimes ($0.52$) both stay well above the
generic wall, and the correctly-specified reference ($0.67$) matches the main
comparison. Its advantage erodes when the state count is wrong: on a three-state
truth fitted by a two-state model it falls to $0.35$, though it still leads the
conditional-variance estimator ($0.27$). Thus a matched model degrades gracefully
under emission and separation error but requires roughly the right latent
dimension.

\begin{table}[ht]
\centering
\small
\caption{Robustness of the matched state-space estimator to misspecification
(median Spearman vs the \emph{exact} $\It$ of the \emph{true} process, $200$ paths,
$T=400$, with $95\%$ bootstrap intervals for \texttt{ssm}). Only the first row is
correctly specified; the rest fit a two-state Gaussian model to a process outside
its family. The generic estimators (A, cv) stay at the wall throughout.}
\label{tab:misspec}
\begin{tabular}{lrrc}
\toprule
True process (fit: 2-state Gaussian) & A & cv & ssm [95\% CI]\\
\midrule
Correct (2-state Gaussian)            & $-0.02$ & $-0.14$ & $\mathbf{0.67}$ {\scriptsize$[0.65,0.69]$}\\
Heavy-tailed emissions ($t_4$)        & $\phantom{-}0.04$ & $\phantom{-}0.00$ & $\mathbf{0.70}$ {\scriptsize$[0.67,0.72]$}\\
Weakly separated regimes ($|\mu|=0.5$)& $-0.01$ & $-0.30$ & $\mathbf{0.52}$ {\scriptsize$[0.42,0.58]$}\\
Under-specified (3-state truth)       & $\phantom{-}0.11$ & $\phantom{-}0.27$ & $\mathbf{0.35}$ {\scriptsize$[0.33,0.37]$}\\
\bottomrule
\end{tabular}
\end{table}

\section{Proof of the Block-Bootstrap Consistency Theorem}
\label{supp:bbproof}

We prove Theorem~\ref{thm:bbconsistency}. Write $\bar S_t=\tfrac1w\sum_{s\in W_t}
\hat D_s^2$ for the window second moment of the estimated revisions, and
$\bar S_t^0=\tfrac1w\sum_{s\in W_t}D_s^2$ for the same with the true revisions
$D_s=m_s-m_{s-1}$.

\emph{Step 1 (the bootstrap reproduces the window mean).} The bootstrap statistic
$(\Ihat_{t,T}^{\mathrm{boot}})^2$ is the circular block-bootstrap estimate of the
mean functional $\bar S_t$. For a sample mean the circular block bootstrap is
mean-unbiased conditional on the data, $\E^*[(\Ihat_{t,T}^{\mathrm{boot}})^2]=\bar S_t$,
with conditional variance $O(\ell/w)\to0$ under (A2)
\citep{kunsch1989jackknife,politis1994stationary}; letting $B\to\infty$ removes
the Monte-Carlo error, so $(\Ihat_{t,T}^{\mathrm{boot}})^2=\bar S_t+o_p(1)$. Thus
the bootstrap contributes a valid confidence band but, as a \emph{point} estimate,
only reproduces the window mean.

\emph{Step 2 (the predictor error is negligible).} By Lemma~\ref{lem:decomp}
applied termwise, $\hat D_s=D_s+(\varepsilon_s-\varepsilon_{s-1})$, whence
\[
  |\bar S_t-\bar S_t^0|
  \le\tfrac1w\textstyle\sum_{s\in W_t}\bigl(2|D_s|+|\Delta\varepsilon_s|\bigr)\,|\Delta\varepsilon_s|,
  \qquad \Delta\varepsilon_s:=\varepsilon_s-\varepsilon_{s-1},
\]
and by the Cauchy--Schwarz inequality with (A1) and (A3) --- writing
$\xi_T:=\max_{s\in W_t}\norm{\varepsilon_s}_{L^2}=o(1)$ for the predictor error
bound of (A3) --- $\E|\bar S_t-\bar S_t^0|\le C(\xi_T+\xi_T^2)\to0$. Hence
$\bar S_t=\bar S_t^0+o_p(1)$.

\emph{Step 3 (window mean of true squares $\to G(u)$).} Split
$\bar S_t^0=\tfrac1w\sum_{s\in W_t}g_s+\tfrac1w\sum_{s\in W_t}\zeta_s$ with
$\zeta_s:=D_s^2-g_s$. Because $g_s=\E[D_s^2\mid\F_{s-1}]$, the sequence
$(\zeta_s,\F_s)$ is a \emph{martingale difference} ($\E[\zeta_s\mid\F_{s-1}]=0$);
the $\zeta_s$ are therefore uncorrelated, and by (A1)
$\Var(\tfrac1w\sum_{s\in W_t}\zeta_s)=\tfrac1{w^2}\sum_{s\in W_t}\E\zeta_s^2=O(1/w)\to0$
--- so this centred term vanishes in $L^2$ \emph{without invoking mixing}. By local
stationarity $g_s=G(s/T)$ with $G$ Lipschitz, so
$\tfrac1w\sum_{s\in W_t}G(s/T)=G(u)+O(w/T)\to G(u)$. Hence
$\bar S_t^0\xrightarrow{p}G(u)$.

\emph{Combining} Steps 1--3,
$(\Ihat_{t,T}^{\mathrm{boot}})^2=\bar S_t+o_p(1)=\bar S_t^0+o_p(1)=G(u)+o_p(1)$;
since $G(u)>0$ and $x\mapsto\sqrt{x}$ is continuous there, the continuous-mapping
theorem gives $\Ihat_{t,T}^{\mathrm{boot}}\xrightarrow{p}\sqrt{G(u)}=\It$.
The bootstrap-validity part~(ii) follows under the additional local-CLT
assumption stated there --- the moving-block bootstrap central limit theorem for
the local mean; \citep{kunsch1989jackknife,lahiri2003resampling} establish the
stationary-case prototype, and (A2) alone is not claimed to yield the
triangular-array version. \hfill$\square$

\emph{Remarks on rigor.} Two ingredients are cited rather than proved here. (i) The
uniform predictor bound (A3) under dependent training data --- for the AR/GARCH
predictors it follows from standard quasi-MLE consistency, but a self-contained
statement is deferred. (ii) The $O(\ell/w)$ block-bootstrap variance in Step~1
rests on a long-run-variance argument that uses the summable mixing of (A2). The
martingale-difference argument of Step~3 --- the crux of the consistency --- is
exact and assumption-light.

\section{Bias Decomposition for the Lightweight Surrogates}
\label{supp:surrbias}

For Surrogate~A, put $r_{s,k}=\tilde D_{s,k}-D_s$ for the predictor error of the
$k$-window running mean. Expanding the square gives the exact identity
\[
  \hat v_{t,A}-v_t
  = \tfrac1w\textstyle\sum_s(D_s^2-v_s)
  + \tfrac1w\sum_s(v_s-v_t)
  + \tfrac2w\sum_s D_s r_{s,k}
  + \tfrac1w\sum_s r_{s,k}^2 .
\]
If $v_s=g(s/T)$ is $L$-Lipschitz and $\{\E(r_{s,k}^2\mid\F_{s-1})\}^{1/2}\le a_{k,T}$
uniformly, conditional Cauchy--Schwarz bounds the last two terms by
$2\sqrt{v_{\max}}\,a_{k,T}$ and $a_{k,T}^2$, so together with the local smoothing
bias
\[
  |\E(\hat v_{t,A}-v_t)|\ \le\ Lw/T + 2\sqrt{v_{\max}}\,a_{k,T} + a_{k,T}^2 .
\]
A local-constant predictor under a Lipschitz mean has the representative order
$a_{k,T}=O(k^{-1/2}+k/T)$, separating finite-window estimation noise from drift.

For Surrogate~B, the AR($p$) revision is $D_t=\phi_1\{X_t-\E(X_t\mid\F_{t-1})\}
=\phi_1\varepsilon_t$ (all lags except the newly observed $X_t$ are
$\F_{t-1}$-measurable), so $\It=|\phi_1|\sigma_t$ and, adding and subtracting
$|\hat\phi_1|\sigma_t$,
\[
  |\Ihat_{t,B}-\It|\ \le\ \sigma_t|\hat\phi_1-\phi_1|
  + |\phi_1|\,|\hat\sigma_t-\sigma_t|
  + |\hat\phi_1-\phi_1|\,|\hat\sigma_t-\sigma_t| .
\]
If $\hat\sigma_t^2$ is a lag-window estimate its error is the Surrogate~A
decomposition with $r_{s,k}=0$ and $D_s$ replaced by the innovation; if it comes
from a conditional-variance model, its error is governed by that model's parameter
and specification error instead.

\section{Estimator Trajectory on One Path}
\label{supp:traj}

Figure~\ref{fig:trajectory} makes the aggregate comparison of Table~\ref{tab:sim}
concrete on a single volatility-driven (AR--GARCH) path. The conditional-variance
estimator tracks the exact $\It$ closely because it models the mechanism that drives
it, whereas the block bootstrap recovers only a smoothed, lagged local average. This
single-path view locates the bootstrap's weakness in a structural lag rather than
sampling noise, and is the mechanism behind its Pareto-dominated point accuracy on
the cost--accuracy frontier of Figure~\ref{fig:frontier}(a) --- consistent with its
role as an interval, not a point, estimator.

\begin{figure}[ht]
\centering
\includegraphics[width=\linewidth]{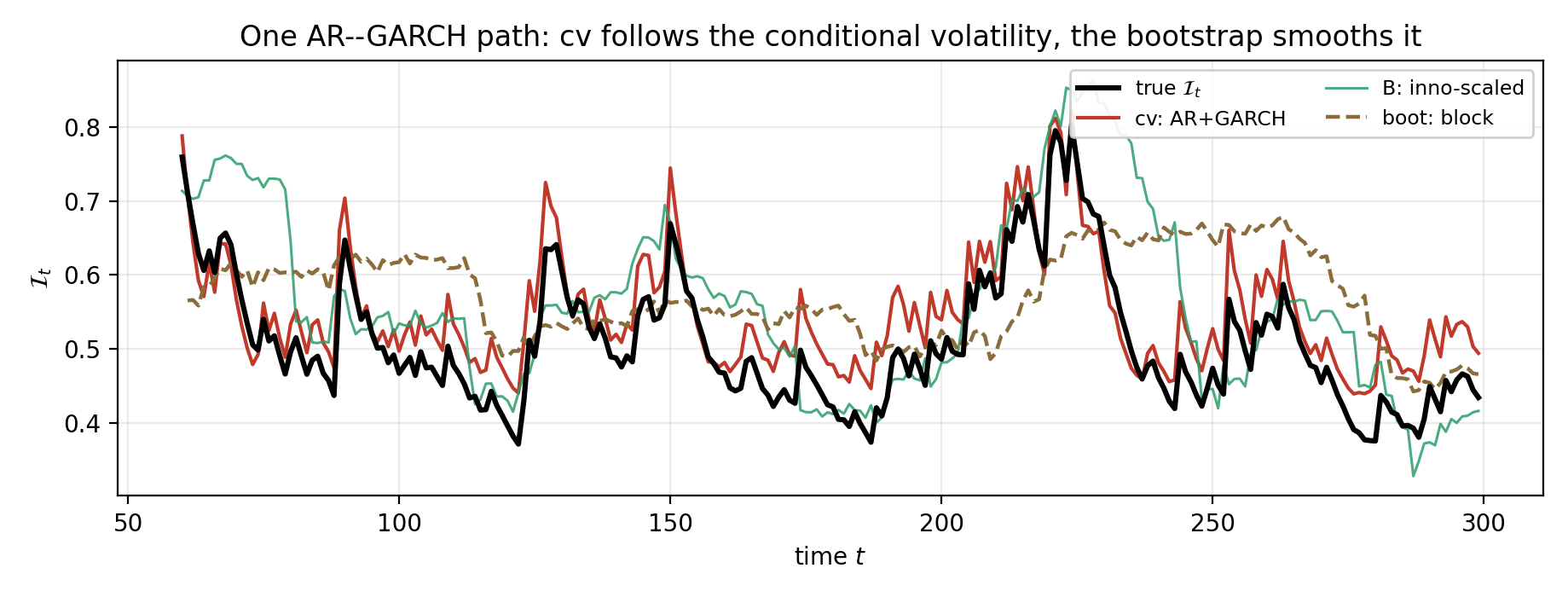}
\caption{One AR--GARCH path: the exact $\It$ (black) against three estimates. The
conditional-variance estimator (red) follows the conditional volatility that drives
$\It$; the innovation-scaled surrogate (green) tracks it more noisily; the block
bootstrap (dashed) recovers only a smoothed, lagged local average --- the mechanism
behind its poor single-path tracking in Figure~\ref{fig:frontier}(a).}
\label{fig:trajectory}
\end{figure}

\section{Adaptive Block-Length Selection}
\label{supp:blocklen}

The block length $\ell$ affects the estimated sampling distribution, not the local
point estimate, so it should \emph{not} be chosen by minimising point-estimate MSE
over $\ell$ (which is essentially unchanged once the Monte-Carlo error is
negligible). We therefore treat $\ell$ as a \emph{sensitivity} parameter rather than reporting a
single selected value: it is swept over a stable grid $1\le\ell\le c\,w^{1/2}$,
assessed by the long-run variance of $Y_s=\hat D_s^2$ within each window, with
interval length reported across the grid and the point estimate unchanged
throughout. If an automatic choice is wanted, any ``minimum-volatility''
selector must be justified as a coverage or long-run-variance criterion, not a
point-MSE one.

\section{Reproducibility and Computational Details}
\label{supp:repro}

\emph{Two reproducibility standards, because two kinds of claim appear in the
paper. Statistical results} (bias, variance, MSE, and the gate correlations) are
pinned by seeds and are machine-independent: every stochastic component derives its
seed from a single master seed via a stable hash of a descriptive label, so any one
configuration reproduces in isolation. These numbers should not move across
machines.

\emph{Cost results} are different. The central framing is a cost-ordered menu, and
the ordering is justified by wall-clock timings, which are inherently
machine-dependent. We therefore state the hardware wherever a timing appears (all
timings on a 12-core Intel Xeon~W-3235 at $3.30$\,GHz, $192$\,GB RAM, macOS;
Python/PyTorch/NumPy versions and BLAS backends recorded in the released run
manifest). Each reported cost is the median wall-clock time over repeated timings,
with the per-run spread recorded in the manifest. The expanded $200$-path
comparison uses process-level parallelism with
one math-library thread per worker; the per-series costs reported in
Table~\ref{tab:sim} are timed \emph{serially}, without that parallel load, so they
are not inflated by contention. A referee on a different machine will reproduce the
\emph{ordering} and the relative costs, not the absolute seconds.

\textbf{Honest limit of seeding.} Identical seeds give bitwise-identical output only
on the same architecture and thread count: floating-point reduction order differs
across BLAS backends (here PyTorch links MKL, NumPy links OpenBLAS), so
cross-machine agreement is to numerical tolerance, not bit-for-bit. Thread counts
are pinned for this reason. All code, random seeds, the exact synthetic series, the
trained network weights (or equivalently the seeds that regenerate them bit-for-bit
on the same machine), and the machine manifest are released in a public repository,
so every reported number can be reproduced and independently re-examined.

\section{Submission-Level Scaling Study}
\label{supp:scaling}

Table~\ref{tab:scaling} in the main paper reports median Spearman across series
lengths; here we give the full detail behind it --- the $95\%$ Monte-Carlo
intervals (Table~\ref{tab:s8spearman}) and the calibrated relative RMSE
(Table~\ref{tab:s8rmse}) --- over $200$ independent paths at each
$T\in\{200,400,800,1600,3200\}$ for the three tracking processes. Two patterns are
worth isolating. First, the \emph{matched} models are consistent: on the
volatility-driven AR--GARCH the conditional-variance estimator's calibrated RMSE
falls from $0.23$ at $T=200$ to $0.05$ at $T=3200$, and on the state-driven
processes the state-space estimator's falls from $\approx0.6$--$0.7$ to
$0.14$--$0.16$, with Spearman rising toward one. Second, the \emph{generic}
estimators are floored: on the HMM and RS-AR every lag-only estimator keeps a
calibrated RMSE indistinguishable from $1.00$ and a Spearman within Monte-Carlo
error of zero \emph{at all $T$} --- more data does not help, the empirical
signature of the lag-only floor of Corollary~\ref{cor:dichotomy}(ii). The one
non-monotone entry, the state-space estimator on AR--GARCH at $T=200$ (Spearman
$0.40$), is the small-sample cost of fitting a four-parameter-per-regime switching
model to a short series; it is resolved by $T=400$ ($0.91$) and is why the cheap
conditional-variance model, not the state-space model, is the recommendation for
volatility-driven $\It$.

\begin{table}[ht]
\centering
\small
\caption{Scaling of tracking accuracy: median Spearman vs the exact $\It$ over
$200$ paths, with $95\%$ percentile-bootstrap intervals, at each series length.}
\label{tab:s8spearman}
\resizebox{\textwidth}{!}{%
\begin{tabular}{llccccc}
\toprule
Process & Estimator & $T{=}200$ & $400$ & $800$ & $1600$ & $3200$\\
\midrule
 \multirow{5}{*}{AR--GARCH} & A: windowed & $0.57$ {\scriptsize$[0.53,0.62]$} & $0.62$ {\scriptsize$[0.60,0.64]$} & $0.65$ {\scriptsize$[0.62,0.66]$} & $0.65$ {\scriptsize$[0.64,0.66]$} & $0.66$ {\scriptsize$[0.65,0.66]$} \\
  & B: inno-scaled & $0.79$ {\scriptsize$[0.77,0.81]$} & $0.81$ {\scriptsize$[0.79,0.81]$} & $0.83$ {\scriptsize$[0.83,0.84]$} & $0.84$ {\scriptsize$[0.83,0.84]$} & $0.84$ {\scriptsize$[0.84,0.85]$} \\
  & cv: GARCH & $0.97$ {\scriptsize$[0.96,0.98]$} & $0.98$ {\scriptsize$[0.98,0.99]$} & $0.99$ {\scriptsize$[0.99,0.99]$} & $1.00$ {\scriptsize$[1.00,1.00]$} & $1.00$ {\scriptsize$[1.00,1.00]$} \\
  & boot: block & $0.44$ {\scriptsize$[0.41,0.47]$} & $0.51$ {\scriptsize$[0.50,0.53]$} & $0.55$ {\scriptsize$[0.54,0.57]$} & $0.57$ {\scriptsize$[0.56,0.58]$} & $0.59$ {\scriptsize$[0.59,0.60]$} \\
  & ssm: state-space & $0.40$ {\scriptsize$[0.29,0.56]$} & $0.91$ {\scriptsize$[0.88,0.93]$} & $0.97$ {\scriptsize$[0.96,0.97]$} & $0.98$ {\scriptsize$[0.97,0.98]$} & $0.98$ {\scriptsize$[0.98,0.98]$} \\
\midrule
 \multirow{5}{*}{HMM} & A: windowed & $-0.00$ {\scriptsize$[-0.02,0.01]$} & $-0.01$ {\scriptsize$[-0.02,-0.00]$} & $-0.01$ {\scriptsize$[-0.01,-0.00]$} & $-0.01$ {\scriptsize$[-0.02,-0.01]$} & $-0.01$ {\scriptsize$[-0.02,-0.01]$} \\
  & B: inno-scaled & $-0.01$ {\scriptsize$[-0.03,-0.00]$} & $-0.03$ {\scriptsize$[-0.04,-0.02]$} & $-0.02$ {\scriptsize$[-0.03,-0.02]$} & $-0.03$ {\scriptsize$[-0.03,-0.02]$} & $-0.02$ {\scriptsize$[-0.03,-0.02]$} \\
  & cv: GARCH & $-0.15$ {\scriptsize$[-0.19,-0.12]$} & $-0.15$ {\scriptsize$[-0.19,-0.11]$} & $-0.13$ {\scriptsize$[-0.16,-0.11]$} & $-0.12$ {\scriptsize$[-0.15,-0.11]$} & $-0.12$ {\scriptsize$[-0.13,-0.11]$} \\
  & boot: block & $-0.01$ {\scriptsize$[-0.03,0.01]$} & $-0.01$ {\scriptsize$[-0.02,-0.00]$} & $-0.02$ {\scriptsize$[-0.03,-0.01]$} & $-0.02$ {\scriptsize$[-0.02,-0.01]$} & $-0.02$ {\scriptsize$[-0.02,-0.01]$} \\
  & ssm: state-space & $0.60$ {\scriptsize$[0.58,0.63]$} & $0.65$ {\scriptsize$[0.64,0.67]$} & $0.71$ {\scriptsize$[0.70,0.73]$} & $0.76$ {\scriptsize$[0.75,0.78]$} & $0.80$ {\scriptsize$[0.79,0.82]$} \\
\midrule
 \multirow{5}{*}{RS-AR} & A: windowed & $0.11$ {\scriptsize$[0.08,0.14]$} & $0.14$ {\scriptsize$[0.11,0.16]$} & $0.14$ {\scriptsize$[0.13,0.15]$} & $0.15$ {\scriptsize$[0.14,0.16]$} & $0.15$ {\scriptsize$[0.14,0.15]$} \\
  & B: inno-scaled & $0.01$ {\scriptsize$[-0.02,0.03]$} & $0.00$ {\scriptsize$[-0.01,0.01]$} & $0.01$ {\scriptsize$[0.00,0.02]$} & $-0.00$ {\scriptsize$[-0.01,0.01]$} & $0.00$ {\scriptsize$[-0.00,0.01]$} \\
  & cv: GARCH & $0.12$ {\scriptsize$[0.10,0.15]$} & $0.15$ {\scriptsize$[0.13,0.16]$} & $0.14$ {\scriptsize$[0.12,0.15]$} & $0.15$ {\scriptsize$[0.14,0.15]$} & $0.17$ {\scriptsize$[0.15,0.17]$} \\
  & boot: block & $-0.02$ {\scriptsize$[-0.06,0.04]$} & $-0.01$ {\scriptsize$[-0.03,0.02]$} & $0.01$ {\scriptsize$[-0.00,0.03]$} & $-0.00$ {\scriptsize$[-0.01,0.01]$} & $0.00$ {\scriptsize$[-0.01,0.01]$} \\
  & ssm: state-space & $0.60$ {\scriptsize$[0.53,0.67]$} & $0.82$ {\scriptsize$[0.79,0.85]$} & $0.94$ {\scriptsize$[0.93,0.95]$} & $0.97$ {\scriptsize$[0.97,0.98]$} & $0.99$ {\scriptsize$[0.98,0.99]$} \\
\bottomrule
\end{tabular}}
\end{table}

\begin{table}[ht]
\centering
\small
\caption{Scaling of the \emph{per-path oracle-calibrated} relative RMSE
(median over $200$ paths) versus the exact $\It$. This is a shape diagnostic, not
the independently calibrated metric reported in main-paper Table~\ref{tab:sim}.
Matched-model error falls with $T$; generic-estimator error stays near one on the
state-driven processes.}
\label{tab:s8rmse}
\begin{tabular}{llccccc}
\toprule
Process & Estimator & $T{=}200$ & $400$ & $800$ & $1600$ & $3200$\\
\midrule
 \multirow{5}{*}{AR--GARCH} & A: windowed & $0.81$ & $0.78$ & $0.76$ & $0.74$ & $0.73$ \\
  & B: inno-scaled & $0.60$ & $0.58$ & $0.54$ & $0.53$ & $0.52$ \\
  & cv: GARCH & $0.23$ & $0.16$ & $0.10$ & $0.08$ & $0.05$ \\
  & boot: block & $0.90$ & $0.86$ & $0.82$ & $0.81$ & $0.79$ \\
  & ssm: state-space & $0.88$ & $0.55$ & $0.46$ & $0.43$ & $0.44$ \\
\midrule
 \multirow{5}{*}{HMM} & A: windowed & $1.00$ & $1.00$ & $1.00$ & $1.00$ & $1.00$ \\
  & B: inno-scaled & $1.00$ & $1.00$ & $1.00$ & $1.00$ & $1.00$ \\
  & cv: GARCH & $1.00$ & $1.00$ & $1.00$ & $1.00$ & $1.00$ \\
  & boot: block & $1.00$ & $1.00$ & $1.00$ & $1.00$ & $1.00$ \\
  & ssm: state-space & $0.58$ & $0.39$ & $0.30$ & $0.23$ & $0.14$ \\
\midrule
 \multirow{5}{*}{RS-AR} & A: windowed & $0.99$ & $0.99$ & $0.99$ & $0.99$ & $0.99$ \\
  & B: inno-scaled & $0.99$ & $1.00$ & $1.00$ & $1.00$ & $1.00$ \\
  & cv: GARCH & $0.97$ & $0.96$ & $0.96$ & $0.96$ & $0.96$ \\
  & boot: block & $0.99$ & $1.00$ & $1.00$ & $1.00$ & $1.00$ \\
  & ssm: state-space & $0.73$ & $0.54$ & $0.31$ & $0.23$ & $0.16$ \\
\bottomrule
\end{tabular}
\end{table}

\section{Uncalibrated and Independently Calibrated Error}
\label{supp:rawmetrics}

The main paper's Table~\ref{tab:sim} scores point accuracy by an \emph{independently
calibrated} RMSE --- one affine map fitted on half of the $200$ paths and scored on
the disjoint half (a single $100/100$ split, not two-fold cross-fitting; the value
reported is the median over the $100$ held-out paths) --- whereas the scaling
table~\ref{tab:s8rmse} uses a
\emph{per-path best-affine} (oracle) calibrated RMSE. The latter isolates
trajectory \emph{shape} but gives each estimate oracle access to the target's
per-path scale and location, so it can flatter an estimator whose shape is right
but whose scale is not. Table~\ref{tab:rawmetrics} recomputes the full battery on
the \emph{same} $200$ paths and exact $\It$ at $T=400$: its Spearman and
independently-calibrated columns reproduce Table~\ref{tab:sim}, and its
oracle-calibrated column reproduces the $T=400$ entries of Table~\ref{tab:s8rmse}
(identical seeding and burn-in), which validates the added columns. The battery
adds the raw (uncalibrated) relative RMSE, normalised both by
$\sigma_{\It}=\mathrm{sd}(\It)$ --- comparable to the calibrated column --- and by
$\mathrm{rms}(\It)$, a level-relative error; the signed level error
$(\overline{\Ihat}-\overline{\It})/\overline{\It}$; and the independently
calibrated RMSE described above.

Two readings follow, both honest. \emph{The matched models are level-unbiased and
shape-accurate.} On AR--GARCH the conditional-variance estimator, and on the
state-driven processes the state-space estimator, carry a signed level error of at
most a few percent (last column) and a small level-relative raw error
($\mathrm{raw}/\mathrm{rms}\,\It\approx0.10$--$0.16$), matching their high Spearman:
they recover both the \emph{level} and the \emph{shape} of $\It$.
\emph{But the per-path calibrated RMSE is optimistic.} For the
conditional-variance estimator on AR--GARCH the oracle-calibrated error
($0.16$, $\sigma$-normalised) is far below both the uncalibrated ($0.55$) and the
\emph{independently} calibrated ($0.53$) values --- a single affine map learned on
separate paths recovers almost none of the oracle gain. The estimator's robust,
un-flattered virtues are therefore its rank tracking (Spearman $0.98$) and its
near-zero level bias; the small calibrated RMSE itself reflects per-path scale
fitting and should not be read as an out-of-sample pointwise accuracy. On the
state-driven processes every generic estimator sits at $\mathrm{raw}/\sigma_{\It}>1$
with a $20$--$46\%$ level error, and only the matched state-space estimator is
level-accurate, so the ``wall'' of the main text is present in the raw battery too
and is not an artifact of calibration. Finally, the AR(1) level check (bottom
block) exposes what the calibrated metric hides: the model-free windowed
Surrogate~A underestimates the constant $\It$ by about one half, whereas the
model-based estimators are within a few percent.

\begin{table}[ht]
\centering
\small
\caption{Uncalibrated error battery on the \emph{same} $200$ paths and exact
$\It$ at $T=400$ as Table~\ref{tab:sim}. ``cal.\ RMSE (oracle)'' is the
per-path best-affine shape diagnostic also used in Table~\ref{tab:s8rmse};
``cal.\ RMSE (indep.)'' is the out-of-sample metric reported in the main paper,
learning one affine map on $100$ paths and scoring it on the disjoint $100$ (median
over the held-out half); ``raw
RMSE'' is uncalibrated, normalised by $\sigma_{\It}=\mathrm{sd}(\It)$ and by
$\mathrm{rms}(\It)$; ``level error'' is the signed relative level bias
$(\overline{\Ihat}-\overline{\It})/\overline{\It}$. All summaries are medians over
the $200$ paths, except cal.\ RMSE (indep.), which is the median over the $100$
held-out paths. Boldface marks the matched estimator per process. AR(1) is a level check
($\It$ constant $\Rightarrow$ Spearman and $\sigma$-normalised errors undefined).}
\label{tab:rawmetrics}
\resizebox{\textwidth}{!}{%
\begin{tabular}{llcccccc}
\toprule
 & & & \multicolumn{2}{c}{cal.\ RMSE} & \multicolumn{2}{c}{raw RMSE} & level\\
\cmidrule(lr){4-5}\cmidrule(lr){6-7}
Process & Estimator & Spearman & oracle & indep. & $/\sigma_{\It}$ & $/\mathrm{rms}\,\It$ & error\\
\midrule
\multirow{5}{*}{AR--GARCH}
 & A: windowed      & $0.62$ & $0.78$ & $0.80$ & $2.26$ & $0.40$ & $-0.37$\\
 & B: inno-scaled   & $0.81$ & $0.58$ & $0.64$ & $0.87$ & $0.16$ & $-0.03$\\
 & \textbf{cv}: GARCH & $\mathbf{0.98}$ & $\mathbf{0.16}$ & $0.53$ & $0.55$ & $0.10$ & $-0.02$\\
 & boot: block      & $0.51$ & $0.86$ & $0.89$ & $1.04$ & $0.19$ & $-0.03$\\
 & ssm: state-space & $0.91$ & $0.55$ & $0.79$ & $0.91$ & $0.16$ & $-0.03$\\
\midrule
\multirow{5}{*}{HMM}
 & A: windowed      & $-0.01$ & $1.00$ & $1.00$ & $2.44$ & $0.40$ & $-0.35$\\
 & B: inno-scaled   & $-0.03$ & $1.00$ & $1.00$ & $1.92$ & $0.32$ & $+0.22$\\
 & cv: GARCH        & $-0.15$ & $1.00$ & $1.00$ & $1.74$ & $0.28$ & $+0.23$\\
 & boot: block      & $-0.01$ & $1.00$ & $1.00$ & $1.78$ & $0.29$ & $+0.22$\\
 & \textbf{ssm}: state-space & $\mathbf{0.65}$ & $\mathbf{0.39}$ & $0.55$ & $0.59$ & $0.10$ & $-0.02$\\
\midrule
\multirow{5}{*}{RS-AR}
 & A: windowed      & $0.14$ & $0.99$ & $1.00$ & $2.40$ & $0.51$ & $-0.46$\\
 & B: inno-scaled   & $0.00$ & $1.00$ & $1.00$ & $1.26$ & $0.26$ & $+0.06$\\
 & cv: GARCH        & $0.15$ & $0.96$ & $0.99$ & $1.06$ & $0.22$ & $+0.07$\\
 & boot: block      & $-0.01$ & $1.00$ & $1.00$ & $1.15$ & $0.24$ & $+0.06$\\
 & \textbf{ssm}: state-space & $\mathbf{0.82}$ & $\mathbf{0.54}$ & $0.62$ & $0.73$ & $0.16$ & $-0.02$\\
\midrule
\multirow{5}{*}{\shortstack[l]{AR(1)\\(level check)}}
 & A: windowed      & --- & --- & --- & --- & $0.51$ & $-0.50$\\
 & B: inno-scaled   & --- & --- & --- & --- & $0.15$ & $-0.03$\\
 & cv: GARCH        & --- & --- & --- & --- & $0.06$ & $-0.02$\\
 & boot: block      & --- & --- & --- & --- & $0.10$ & $-0.03$\\
 & ssm: state-space & --- & --- & --- & --- & $0.10$ & $-0.04$\\
\bottomrule
\end{tabular}}
\end{table}

\section{Proof of the Matched-Model Achievability Result}
\label{supp:achieve}

We give a detailed proof of Proposition~\ref{prop:achieve}; the single step not
expanded in full --- geometric stability of the tangent filter in Step~4 --- is a
standard filter-differentiability property under the regularity of Douc et
al.\ \citep{douc2004asymptotic}. Throughout,
$\theta$ collects the transition matrix $P$ and the regime-specific parameters
$\{\mu_k,\phi_k,\sigma_k\}_{k=1}^{K}$ of the $K$-state Markov-switching
autoregression
\begin{equation}
  X_t=\mu_{S_t}+\phi_{S_t}X_{t-1}+\sigma_{S_t}\varepsilon_t,\qquad
  \varepsilon_t\stackrel{\text{iid}}{\sim}\mathcal N(0,1),
  \label{eq:msar}
\end{equation}
with $(S_t)$ a stationary Markov chain on $\{1,\dots,K\}$ with transition matrix
$P$, independent of $(\varepsilon_t)$. Let $\Theta$ be a compact parameter set and
$\theta_0\in\operatorname{int}\Theta$ the data-generating value. We use four
regularity conditions, all standard for Markov-switching autoregressions; they
hold for the two-state designs of the main paper, whose transition matrix
$\big[\begin{smallmatrix}0.9&0.1\\0.1&0.9\end{smallmatrix}\big]$ is strictly
positive.

\begin{itemize}[leftmargin=1.4em]
\item[(SS1)] \emph{Stationarity and moments.} $P_0$ is irreducible and
aperiodic with stationary distribution $\pi_0$ of full support, and the
Markov-modulated recursion \eqref{eq:msar} has negative top Lyapunov exponent
(e.g.\ $\sum_k\pi_{0,k}\phi_{0,k}^2<1$). Then $(S_t,X_t)$ is stationary and
geometrically ergodic; we assume in addition $\E X_t^4<\infty$. The
forecast-revision variance $v_t:=\It^2$ is then a stationary, integrable process of
the form $v_t=O(1+X_{t-1}^2)$ --- almost surely finite but \emph{not} uniformly
bounded --- and we do not require it bounded away from zero (the final square-root
step uses $|\sqrt a-\sqrt b|\le\sqrt{|a-b|}$).
\item[(SS2)] \emph{Identifiability.} $\theta_0$ is identifiable from the law of
$(X_t)$ up to permutation of the state labels.
\item[(SS3)] \emph{MLE regularity.} The emission densities are positive and
twice continuously differentiable in $\theta$ on the compact $\Theta$, and the
conditions of Douc et al. \citep{douc2004asymptotic} for strong consistency of
the maximum
likelihood estimator in autoregressive models with Markov regime hold.
\item[(SS4)] \emph{Filter minorization.} $\min_{j,k}(P_0)_{jk}\ge\epsilon_->0$,
and hence $\min_{j,k}(P)_{jk}>0$ on a neighbourhood $U\ni\theta_0$.
\end{itemize}

Let $\hat\theta_T=\arg\max_{\theta\in\Theta}\ell_T(\theta)$ be the offline
(full-sample) MLE, $\ell_T$ the exact log-likelihood evaluated by the Hamilton
filter, and fix $u\in(0,1)$, $t=\lfloor uT\rfloor$.

\paragraph{Step 1: strong consistency of the MLE}
Under (SS1)--(SS3), Douc et al. \citep{douc2004asymptotic} give
$\hat\theta_T\to\theta_0$ almost surely, up to a permutation of the state labels.
Every quantity below --- the filtered forecast-revision variance $v_t$ --- is a
functional of the \emph{observable} predictive law and is therefore invariant
under relabelling of the latent states; fixing a labelling we treat
$\hat\theta_T\to_{\mathrm{a.s.}}\theta_0$.

\paragraph{Step 2: the filtered functional}
For a parameter $\theta$ and an initial law $\xi_0$ on $\{1,\dots,K\}$, the
Hamilton prediction filter is
\begin{equation}
  \alpha_t(\theta)=P^\top\xi_{t-1}(\theta),\qquad
  \xi_t(\theta)\propto \operatorname{diag}\!\big\{g_\theta(X_t\mid X_{t-1},k)\big\}_k\,
  P^\top\xi_{t-1}(\theta),
  \label{eq:filter}
\end{equation}
with $g_\theta(x\mid x',k)=\sigma_k^{-1}\varphi\big((x-\mu_k-\phi_k x')/\sigma_k\big)$
the Gaussian emission density. The forecast-revision variance is the functional
\begin{equation}
  v_t(\theta;\xi_0)
  =\Var_\theta\!\big(\E_\theta[X_{t+1}\mid\F_t]\mid\F_{t-1}\big)
  =\Phi\big(\alpha_t(\theta),X_{t-1};\theta\big),
  \label{eq:vfunctional}
\end{equation}
where $\Phi$ maps the predictive state distribution $\alpha_t$, the previous
observation $X_{t-1}$, and $\theta$ to the variance of
$m_\theta(X_t)=\E_\theta[X_{t+1}\mid\F_{t-1},X_t]$ under the predictive mixture
$X_t\mid\F_{t-1}\sim\sum_k\alpha_{t,k}\,\mathcal N(\mu_k+\phi_k X_{t-1},\sigma_k^2)$.
A finite-Gaussian-mixture computation shows $\Phi$ is $C^1$ in $(\alpha,\theta)$
and locally Lipschitz, with constants uniform on compacts and, in $\alpha$, of the
form $L_\Phi(X_{t-1})=O(1+X_{t-1}^2)$.

\paragraph{Step 3: geometric forgetting of the filter}
The normalised update \eqref{eq:filter} is $\xi\mapsto \operatorname{diag}(g)\,
P_\theta^\top\xi$ up to scaling. In the Hilbert projective metric $d_H$ on the
positive cone, a positive diagonal map is an isometry ---
$(\!\operatorname{diag}(g)\xi)_i/(\!\operatorname{diag}(g)\xi')_i=\xi_i/\xi'_i$ ---
so the contraction is governed by the positive matrix $P_\theta^\top$ alone. By the
classical Birkhoff--Hopf theorem the map at $\theta$ has contraction coefficient
$\tau(P_\theta)=\frac{1-\sqrt{\varpi_\theta}}{1+\sqrt{\varpi_\theta}}<1$, where
$\varpi_\theta=\min_{i,j,k,l}\frac{(P_\theta)_{ik}(P_\theta)_{jl}}{(P_\theta)_{jk}(P_\theta)_{il}}>0$;
by (SS4) and continuity of $\theta\mapsto\tau(P_\theta)$ we shrink $U$ so that
$\rho:=\sup_{\theta\in U}\tau(P_\theta)<1$. Hence, for any two initialisations
$\xi_0,\xi_0'$ and all $\theta\in U$,
\begin{equation}
  d_H\!\big(\xi_t(\theta;\xi_0),\xi_t(\theta;\xi_0')\big)\le \rho^{\,t}\,
  d_H(\xi_0,\xi_0'),\qquad
  \|\xi_t(\theta;\xi_0)-\xi_t(\theta;\xi_0')\|_{\mathrm{TV}}\le C\,\rho^{t},
  \label{eq:forget}
\end{equation}
with $\rho<1$ and $C<\infty$ (the second bound is the standard consequence of the
first). This forgetting is \emph{uniform} in $\theta\in U$ and
in the data, and requires no moment condition, because the emission factor drops
out of $d_H$. Writing $\bar\xi_t(\theta)$ for the stationary filter (the a.s.\
limit as the initialisation recedes to $-\infty$) and
$\bar v_t(\theta)=\Phi(\bar\alpha_t(\theta),X_{t-1};\theta)$, and combining
\eqref{eq:forget} with the $\alpha$-Lipschitz bound of Step~2,
\[
  |v_t(\theta;\xi_0)-\bar v_t(\theta)|\le L_\Phi(X_{t-1})\,C\rho^{\,t}.
\]
Since $\E L_\Phi(X_{t-1})<\infty$ (SS1) and $\sum_t\rho^t\E L_\Phi(X_{t-1})<\infty$,
Borel--Cantelli gives $L_\Phi(X_{t-1})\rho^t\to0$ a.s., so the initialisation error
is $o_{\mathrm{a.s.}}(1)$ as $t=\lfloor uT\rfloor\to\infty$.

\paragraph{Step 4: uniform Lipschitz continuity in $\theta$, and plug-in}
Differentiating \eqref{eq:filter} in $\theta$, the tangent filter
$\partial_\theta\xi_t$ solves a linear recursion whose homogeneous part is the same
positive map as in \eqref{eq:forget}; it is therefore geometrically stable and,
driven by the $C^1$ (SS3) and $L^2$-bounded (SS1, via $\E X_t^4<\infty$) score
increments, has stationary integrable norm. (A fully explicit treatment of this
inhomogeneous derivative recursion and the domination of the parameter-derivative
terms is standard for Markov-switching filters under the regularity of Douc et
al.\ \citep{douc2004asymptotic}; we take it as given here.) With the smoothness of $\Phi$ in $(\alpha,\theta)$ from Step~2,
$\theta\mapsto\bar v_t(\theta)$ is Lipschitz on $U$ with a stationary constant
$\Lambda_t$ satisfying $\sup_t\E\Lambda_t<\infty$, hence $\Lambda_t=O_p(1)$.
Therefore
\begin{equation}
  |\bar v_t(\hat\theta_T)-\bar v_t(\theta_0)|\le
  \Lambda_t\,\|\hat\theta_T-\theta_0\|
  = O_p(1)\cdot o_{\mathrm{a.s.}}(1)=o_p(1),
  \label{eq:plugin}
\end{equation}
the last equality because $\|\hat\theta_T-\theta_0\|\to0$ a.s.\ (Step~1) while
$\Lambda_{\lfloor uT\rfloor}=O_p(1)$.

\paragraph{Step 5: correct specification}
At $\theta_0$ the stationary Hamilton filter is exact:
$\bar\xi_{t-1}(\theta_0)=P(S_{t-1}=\cdot\mid\F_{t-1})$, the true filtering
distribution, so $\bar\alpha_t(\theta_0)=P^\top\bar\xi_{t-1}(\theta_0)
=P(S_t=\cdot\mid\F_{t-1})$ is the true predictive state law and,
by \eqref{eq:vfunctional} and Definition~\ref{def:It},
\[
  \bar v_t(\theta_0)=\Var\!\big(\E[X_{t+1}\mid\F_t]\mid\F_{t-1}\big)=\It^2 .
\]
(The gap between the infinite-past filter and the filter started at $t=1$ is the
same $O(\rho^t)$ term already controlled in Step~3.)

\paragraph{Conclusion}
Chaining Steps~3--5,
\[
  v_t(\hat\theta_T;\xi_0)=\bar v_t(\hat\theta_T)+o_{\mathrm{a.s.}}(1)
  =\bar v_t(\theta_0)+o_p(1)=\It^2+o_p(1).
\]
Since $|\sqrt a-\sqrt b|\le\sqrt{|a-b|}$ for $a,b\ge0$, the $o_p(1)$ bound on
$v_t(\hat\theta_T;\xi_0)-\It^2$ transfers directly to
$\Ihat_t^{\mathrm{ssm}}=\sqrt{v_t(\hat\theta_T;\xi_0)}\xrightarrow{p}\It$, with no
positivity floor required. Finally,
$\Ihat_t^{\mathrm{ssm}}$ is a global nonlinear functional of the whole series
(through $\hat\theta_T$ and the full filter recursion), not a convex combination
of $\{D_s^2\}_{s<t}$ with externally fixed weights, so it lies outside
Definition~\ref{def:local} and the floor of Corollary~\ref{cor:dichotomy}(ii)
does not bind. \hfill$\square$

\begin{remark}[Scope: offline versus online]
The estimator is \emph{offline}: $\hat\theta_T$ uses the whole series, so the
statement is one of retrospective ($u$ fixed, $T\to\infty$) consistency. An online
estimator that refits using only $\F_{t-1}$ and filters forward is a distinct
object; the same five steps apply with $\hat\theta_{t-1}$ replacing $\hat\theta_T$
whenever the MLE is consistent on the growing window, but the finite-sample
behaviour differs and we do not analyse it. Correct specification is essential in
Step~5; Section~\ref{supp:misspec} quantifies the deviation of $\bar v_t(\theta_*)$
from $\It^2$ under misspecification, where $\theta_*$ is the Kullback--Leibler
projection of the truth onto the fitted family.
\end{remark}

\section{Proof of the Gaussian Information Identity}
\label{supp:miproof}

We give the term-by-term proof of Proposition~\ref{prop:mi}. Write
$m_s=\E[X_{t+1}\mid\F_s]$, $\sigma^2_{\mathrm{post},t}=\Var(X_{t+1}\mid\F_t)$
(assumed $\F_{t-1}$-measurable), and
$\sigma^2_{\mathrm{prior},t}=\Var(X_{t+1}\mid\F_{t-1})$.

\emph{Step 1 (variance decomposition).} By the law of total variance and the
$\F_{t-1}$-measurability of $\sigma^2_{\mathrm{post},t}$,
\begin{equation}
  \sigma^2_{\mathrm{prior},t}
  =\underbrace{\Var\!\big(\E[X_{t+1}\mid\F_t]\mid\F_{t-1}\big)}_{=\,\It^2}
   +\underbrace{\E[\sigma^2_{\mathrm{post},t}\mid\F_{t-1}]}_{=\,\sigma^2_{\mathrm{post},t}}
  =\It^2+\sigma^2_{\mathrm{post},t}.
  \label{eq:mi-var}
\end{equation}

\emph{Step 2 (the two Gaussian log-densities).} By joint Gaussianity, conditionally
on $\F_{t-1}$,
\begin{align}
  X_{t+1}\mid X_t,\F_{t-1} &\sim \mathcal N\!\big(m_t,\ \sigma^2_{\mathrm{post},t}\big),
  \label{eq:mi-post}\\[2pt]
  X_{t+1}\mid\F_{t-1}      &\sim \mathcal N\!\big(m_{t-1},\ \sigma^2_{\mathrm{prior},t}\big),
  \label{eq:mi-prior}
\end{align}
so their log-density ratio is
\begin{equation}
  \log\frac{p(X_{t+1}\mid X_t,\F_{t-1})}{p(X_{t+1}\mid\F_{t-1})}
  =\tfrac12\log\frac{\sigma^2_{\mathrm{prior},t}}{\sigma^2_{\mathrm{post},t}}
   -\frac{(X_{t+1}-m_t)^2}{2\,\sigma^2_{\mathrm{post},t}}
   +\frac{(X_{t+1}-m_{t-1})^2}{2\,\sigma^2_{\mathrm{prior},t}}.
  \label{eq:mi-ratio}
\end{equation}

\emph{Step 3 (conditional expectation).} The two quadratic terms of
\eqref{eq:mi-ratio} have conditional means
\begin{equation}
  \E\!\big[(X_{t+1}-m_t)^2\mid\F_{t-1}\big]=\sigma^2_{\mathrm{post},t},
  \qquad
  \E\!\big[(X_{t+1}-m_{t-1})^2\mid\F_{t-1}\big]=\sigma^2_{\mathrm{prior},t},
  \label{eq:mi-quad}
\end{equation}
the first by the tower property,
$\E[(X_{t+1}-m_t)^2\mid\F_{t-1}]=\E\!\big[\E[(X_{t+1}-m_t)^2\mid\F_t]\mid\F_{t-1}\big]
=\E[\sigma^2_{\mathrm{post},t}\mid\F_{t-1}]=\sigma^2_{\mathrm{post},t}$. Hence the two
quadratic terms contribute $-\tfrac12+\tfrac12=0$ in conditional mean, and taking
$\E[\,\cdot\mid\F_{t-1}]$ of \eqref{eq:mi-ratio} leaves
\begin{equation}
  \mathcal J_t
  :=\E\!\left[\log\frac{p(X_{t+1}\mid X_t,\F_{t-1})}{p(X_{t+1}\mid\F_{t-1})}\ \bigg|\ \F_{t-1}\right]
  =\tfrac12\log\frac{\sigma^2_{\mathrm{prior},t}}{\sigma^2_{\mathrm{post},t}}
  \overset{\eqref{eq:mi-var}}{=}\tfrac12\log\!\left(1+\frac{\It^2}{\sigma^2_{\mathrm{post},t}}\right),
  \label{eq:mi-final}
\end{equation}
the stated identity. A further expectation over $\F_{t-1}$ gives the conventional
conditional mutual information $I(X_t;X_{t+1}\mid\F_{t-1})=\E[\mathcal J_t]$.

\emph{Necessity of the measurability restriction.} If $\sigma^2_{\mathrm{post},t}$
depends on $X_t$, Step~1 no longer isolates $\It^2$, since
$\E[\sigma^2_{\mathrm{post},t}\mid\F_{t-1}]\ne\sigma^2_{\mathrm{post},t}$ and the
post-update variance mixes with the mean-update variance. A pure variance channel
makes this concrete: for $X_{t+1}=\sqrt{\omega+\alpha X_t^2}\,\eta_{t+1}$ with
$\eta_{t+1}\sim\mathcal N(0,1)$ independent of $\F_t$, one has
$\E[X_{t+1}\mid\F_t]\equiv0$, so $\It\equiv0$; yet observing $X_t$ sharpens the
predictive variance, giving $\mathcal J_t>0$. Thus $\It$ alone determines neither
$\mathcal J_t$ nor $I(X_t;X_{t+1}\mid\F_{t-1})$ once the measurability restriction
fails. \hfill$\square$

\bibliographystyle{elsarticle-num}
\bibliography{P2_references}

\end{document}